\documentclass[twocolumn]{autart}    

\usepackage{graphicx}          

\usepackage{cite}
\usepackage{amsmath,amssymb,amsfonts}
 \newenvironment{proof}{\emph{Proof:}}{\hspace{\stretch{1}}\rule{1ex}{1ex}}  
\usepackage{algorithm, algorithmic}
\usepackage{textcomp}
\usepackage{float}
\usepackage{subfigure}
\usepackage{color}
\usepackage[justification=centering]{caption}
\newtheorem{theorem}{Theorem}

\newtheorem{assumption}{Assumption}

\newtheorem{lemma}{Lemma}

\newtheorem{remark}{Remark}

\usepackage{ragged2e}

\begin{document}

\begin{frontmatter}

\title{Recursive Projection-Free Identification with Binary-Valued Observations\thanksref{footnoteinfo}} 

\thanks[footnoteinfo]{The work is supported by National Natural Science Foundation of China under Grants 62025306, 62303452 and T2293773, CAS Project for Young Scientists in Basic Research under Grant YSBR-008, China Postdoctoral Science Foundation under Grant 2022M720159.Corresponding author: Yanlong Zhao.}

\author[2,3]{Tianning Han}\ead{hantianning23@mails.ucas.ac.cn},    
\author[2]{Ying Wang}\ead{wangying96@amss.ac.cn},               
\author[2,3]{Yanlong Zhao*}\ead{ylzhao@amss.ac.cn}  

 \address[2]{Key Laboratory of Systems and Control, Academy of Mathematics and Systems Science,
 	Chinese Academy of Sciences, Beijing~100190, P.~R.~China}                                             
\address[3]{School of Mathematical Sciences, University of Chinese Academy of Sciences, Beijing~100049, P.~R.~China}             
    


\begin{abstract}                         
This paper is concerned with parameter identification problem for finite impulse response (FIR) systems with binary-valued observations under low computational complexity. Most of the existing algorithms under binary-valued observations rely on projection operators, which leads to a high computational complexity of much higher than $O\left(n^2\right)$. In response, this paper introduces a recursive projection-free identification algorithm that incorporates a specialized cut-off coefficient to fully utilize prior information, thereby eliminating the need for projection operators. The algorithm is proved to be mean square and almost surely convergent. Furthermore, to better leverage prior information, an adaptive accelerated coefficient is introduced, resulting in a mean square convergence rate of $O\left(\frac{1}{k}\right)$, which matches the convergence rate with accurate observations. Inspired by the structure of the $\mathrm{Cram\acute{e}r}$-Rao lower bound, the algorithm can be extended to an information-matrix projection-free algorithm by designing adaptive weight coefficients. This extension is proved to be asymptotically efficient for first-order FIR systems, with simulations indicating similar results for high-order FIR systems. Finally, numerical examples are provided to demonstrate the main results.
\end{abstract}

\end{frontmatter}

\section{Introduction}
\label{sec:introduction}
With the advancement of information technology, biotechnology, and big data, a new class of systems known as set-valued systems has emerged \cite{1245179}. In these systems, the output cannot be accurately measured; only whether they belong to certain known sets can be determined. For instance, due to measurement limitations, sensor network data often result in set-valued observations with a finite number of bits, or even just 1 bit \cite{dargie2010fundamentals, huang2017efficient, tei2014model}. This means that each sensor only indicates whether the measured value exceeds a threshold. In biological systems, we can only ascertain whether a neuron is in an excited or inhibited state \cite{ghysen2003origin}. When the potential is below the threshold, the neuron shows inhibition; otherwise, it shows excitation. Additionally, gas content sensors in the gas and oil industry \cite{sun2004aftertreatment}, switching sensors for shift-by-wire in automotive applications \cite{1383873}, and traffic condition indicators in asynchronous transmission mode networks \cite{schwiebert2001robust} all operate with set-valued observations.

As mentioned above, the increasingly important role of set-valued systems is recognized in practical applications. Unlike traditional systems with accurate observations, set-valued measurements offer minimal information and exhibit high nonlinearity. Consequently, conventional identification methods like least squares, Kalman filtering, and other classical approaches are ill-suited for set-valued systems. Thus, there arises a necessity to develop identification methods tailored for such systems.

Over the past two decades, since the inception of set-valued systems, extensive and profound research has been undertaken, yielding numerous significant results\cite{papadopoulos2001sequential,zhang2013quantized,gustafsson2009statistical,le2010system,godoy2011identification,song2018recursive,ge2017adaptive}. Parameter identification with set-valued observations is one of the most important topics about set-valued systems \cite{1245179}. Plenty of identification algorithms have been proposed and they can be classified into two parts according to whether using projection operators. 

Without using the projection operators, there are a series of achievements
concerning system identification with binary observations. For instance, an empirical measure (EM) method without projection operators is introduced in \cite{1245179} for finite impulse response (FIR) systems with periodic inputs, demonstrating convergence and asymptotic efficiency. In \cite{WANG20071178}, an optimal convex combination estimation of the EM algorithm is constructed under multi-threshold measurements, with proven asymptotic efficiency, but requiring uniform boundedness of the inverse of the distribution function. Guo et al. present an EM algorithm in \cite{guo2015asymptotically} and extend it to systems with quantized inputs and outputs, with available key convergence properties under periodic and full-rank inputs. 
A cyclic projection-free identification algorithm based on binary observations is proposed in \cite{5717798}, though its convergence is demonstrated solely through numerical simulations without theoretical proof. Under conditions of independently and identically distributed inputs, \cite{6426238} introduces a low-storage and low-complexity recursive identification algorithm, establishing theoretical analysis in noise-free systems. Also under independently and identically distributed input conditions, You utilizes the design of an adaptive quantizer to propose a stochastic approximation-type recursive estimation algorithm in \cite{you2015recursive} and establishes the convergence of the estimation and the asymptotic efficiency of the algorithm. 
The algorithms proposed in these works can effectively handle the parameter identification problem of set-valued systems, but some exhibit input constraints and others lack desirable convergence properties.

Progress in identification algorithms employing projection operators with set-valued measurements has also been notable. For instance, \cite{guo2013recursive} presents a recursive projection algorithm, proving its mean square and almost sure convergence under persistent excitation (PE) conditions, with convergence rates of $O(\frac{\log k}{k})$. Further advancements are made in \cite{8623153}, establishing that the mean square convergence rate of the recursive projection algorithm under PE conditions can reach $O\left(\frac{1}{k}\right)$, highlighting its dependency on true parameters rather than prior information. 
Additionally, \cite{wang2022unified} investigates a unified sign-error type algorithm with projection, demonstrating mean square and almost sure convergence to the true parameter, with the mean square convergence rate of estimation error at $O\left(\frac{1}{k}\right)$. An online Quasi-Newton type algorithm with time-varying projection is proposed in \cite{ZHANG2022110158}, with established almost sure convergence under a signal condition weaker than traditional PE conditions. In \cite{10247592}, a two-step Quasi-Newton algorithm with time-varying projection is introduced to enhance identification performance, with proven global convergence and asymptotic normality. Furthermore, \cite{wang2023asymptotically} constructs an information-based projection algorithm, proving its asymptotic efficiency under certain noise constraints. However, despite their impressive convergence properties, algorithms employing projection operators often suffer from high computational complexity.


Based on the above discussion, it can be observed that current identification algorithms for set-valued systems still have significant potential for further improvement. this paper is concerned with the parameter identification problem for FIR systems with binary-valued observations under low computational complexity. The main contributions of this paper can be summarized as follows:



\begin{itemize}
	\item This paper focuses on designing projection-free identification algorithms for FIR systems with binary-valued observations that operate under low computational complexity. It is worth noting that projection operators ensure good convergence properties by fully utilizing the prior information of unknown parameters. Thus, the main challenge is to leverage this prior information while eliminating the need for projection operators. More specifically, we ensure that the estimation error of the algorithm gradually decreases during the iterative process while removing the projection operators. To address this, we design a cut-off coefficient to limit the product of estimate and regressor when it deviates from the bound of their prior information, and an adaptive accelerated coefficient to increase the step size. Consequently, we can eliminate the projection operators while maintaining good convergence properties.
	
	
	\item This paper proposes a recursive projection-free identification (RPFI) algorithm for binary-valued observations, incorporating a cut-off coefficient and an adaptive accelerated coefficient. Compared to the existing algorithms in\cite{guo2013recursive,8623153,guo2014identification, zhang2019asymptotically,wang2022unified,10247592,wang2023asymptotically}, the proposed algorithm eliminates the projection operator, thereby significantly reducing computational complexity. Both the mean square and almost surely convergences of the algorithm are established. Furthermore, the mean square convergence rate is proved to be $O\left(\frac{1}{k}\right)$ for high-order FIR systems, a result that, in \cite{zhang2019asymptotically}, could only be numerically simulated and lacked theoretical proof.
	
	\item
	Inspired by the recursive form of the CR lower bound, an information-matrix projection-free (IMPF) algorithm is constructed based on designing proper weight coefficients of the RPFI algorithm. The IMPF algorithm is proved to be asymptotically efficient for first-order FIR systems, which can not be obtained for the algorithm in \cite{Ke2022RecursiveIO}. By numerical simulations, the asymptotic efficiency is also suitable for high-order FIR systems. This paper also illustrates that the computational complexity of the algorithm proposed is much lower than other asymptotically efficient algorithms in \cite{zhang2019asymptotically,10247592,wang2023asymptotically} by numerical simulations.

\end{itemize}

The remainder of this paper is organized as follows. Section 2 describes the identification problem of FIR system with binary-valued observations. The RPFI algorithm is introduced in Section 3, where its convergence properties are studied. In Section 4, the IMPF algorithm is constructed and it is proved to be asymptotically efficient for first-order FIR systems. Three numerical examples are illustrated in Section 5 to show the main results. Section 6 gives the concluding remarks and related future works.

\section{Problem formulation}
\subsection{System Description}
Consider the FIR system
\begin{gather}\label{M}
	y_k=\phi_k^T\theta+d_k, k\geq0,
\end{gather}
where $\phi_k\in \mathbb{R}^n$, $\theta\in \mathbb{R}^n$, and $d_k\in \mathbb{R}$ are the system input, unknown but constant parameter, and noise at time $k$, respectively. The system output $y_k$ can not be accurately measured, but its binary-valued observation can be obtained, i.e.,
\begin{equation}\label{moxingtwo}
	\begin{aligned}
		s_k=I_{\{y_k\leq C\}}=
		\left\{
		\begin{aligned}
			1&, & y_k\leq C,\\
			0&, & y_k> C.\\
		\end{aligned}
		\right.\\
	\end{aligned}
\end{equation}
where $C$ is a known constant threshold and $I_{\{\cdot\}}$ is the indicator function.

\subsection{Assumption}
In order to proceed our analysis, we introduce some assumptions concerning the unknown parameter, the random noise and the regressor.

\begin{assumption}\label{A1}
	The unknown parameter $\theta$ belongs to a bounded set with the boundedness $\bar{\theta}$, i.e. $\|\theta\|\leq\bar{\theta}$.
\end{assumption}

\newtheorem{assumption4}{\bf Assumption}
\begin{assumption}\label{A2}
	$\{d_k,k\geq1\}$ is an independent and identically distributed (i.i.d.) stochastic normal sequence with zero mean and known covariance $\sigma^2$. The probability distribution function and probability density function of $d_k$ are denoted as $F(\cdot)$ and $f(\cdot)$, respectively.
\end{assumption}

\newtheorem{assumption2}{\bf Assumption}
\begin{assumption}\label{A3}
	The input sequence $\{\phi_k\}$ satisfies $\sup_{k\geq 1}\Vert\phi_k\Vert\triangleq\bar{\phi}<\infty$. Besides, there exists a positive integer N and a positive constant $\delta$ such that 
	\begin{gather*}
		\frac{1}{N}\sum_{i=k}^{k+N-1}\phi_i\phi_i^T\geq\delta^2I_n, \forall k\geq1.
	\end{gather*}
	
\end{assumption}

\begin{remark}\label{n0}
	Assumption 2 can be generalized into the symmetrically distributed noise whose probability distribution function is second-order continuous derivable with mean 0 and known covariance.
\end{remark}

	Actually, if the prior information of unknown parameter $\theta\in\Theta$ is an affine set, the projection operator is equivalent to multiplying by an n-dimensional projection matrix\cite{Boyd_Vandenberghe_2004}, whose computational complexity is $O\left(n^2\right)$. However, $\Theta$ in \cite{guo2013recursive,8623153,guo2014identification,wang2023asymptotically, zhang2019asymptotically,wang2022unified,10247592} is bounded convex compact set, which means the prior information set is much smaller and the projection operator is more complex. Therefore, the computational complexity of using the projection operator is much higher than $O\left(n^2\right)$. Hence this paper focus on designing projection-free identification algorithms for FIR systems with binary-valued observations to reduce the computational complexity.

\section{The RPFI algorithm and its convergence properties}
This section will establish a projection-free identification algorithm with binary valued observations and its convergence properties. The idea of the algorithm design is as follows. As mentioned above, the projection operators bring high computational complexity. But because it make full use of the prior information of unknown parameters, the projection operators ensure good convergence properties. Therefore, an important challenge is how to leverage prior information effectively while reducing the high computational complexity introduced by projection. Through utilizing prior information of unknown parameter, we design a cut-off coefficient 
\begin{align*}
	z_k=
	\left\{
	\begin{aligned}
		&M, & \phi_k^T\hat{\theta}_{k-1}&>M;\\
		&\phi_k^T\hat{\theta}_{k-1}, & \phi_k^T\hat{\theta}_{k-1}&\in[-M,M];\\
		&-M, & \phi_k^T\hat{\theta}_{k-1}&<-M,
	\end{aligned}
	\right.
\end{align*}	
to remove the projection operators and introduce an adaptive accelerated coefficient 
\begin{align*}
	\gamma_k=
	\left\{
	\begin{aligned}
		&1, & |\phi_k^T\hat{\theta}_{k-1}|&\leq M;\\
		&|\phi_k^T\hat{\theta}_{k-1}|, & |\phi_k^T\hat{\theta}_{k-1}|&>M,
	\end{aligned}
	\right.
\end{align*}
to maintain high convergence rate, where $M=\bar{\phi}\bar{\theta}+2$. Therefore,  we propose the following RPFI algorithm.

\begin{algorithm}[htb]
	\caption{The RPFI Algorithm}
	\label{Algorithm}
	For the arbitrary initial value $\hat{\theta}_0\in\mathbb{R}^n$, the algorithm is recursively defined at any $k\geq0$ as follows:
	
	Step 1: Update of the cut-off coefficient:
	\begin{align}\label{eq1.1}
		\begin{aligned}
			z_k=&
			\left\{
			\begin{aligned}
				&M, & \phi_k^T\hat{\theta}_{k-1}&>M;\\
				&\phi_k^T\hat{\theta}_{k-1}, & \phi_k^T\hat{\theta}_{k-1}&\in[-M,M];\\
				&-M, & \phi_k^T\hat{\theta}_{k-1}&<-M,
			\end{aligned}
			\right.
		\end{aligned}
	\end{align}
	and the accelerated coefficient
	\begin{align}\label{eq1.2}
		\begin{aligned}
			\gamma_k=&
			\left\{
			\begin{aligned}
				&1, & |\phi_k^T\hat{\theta}_{k-1}|&\leq M;\\
				&|\phi_k^T\hat{\theta}_{k-1}|, & |\phi_k^T\hat{\theta}_{k-1}|&>M,
			\end{aligned}
			\right.
		\end{aligned}
	\end{align}
	where $M=\bar{\phi}\bar{\theta}+2$.
	
	Step 2: Estimation:
	\begin{align}\label{eq1}
		\left\{
		\begin{aligned}
			\hat{\theta}_{k}=&\hat{\theta}_{k-1}+\frac{\gamma_k\alpha_k\phi_k}{r_k}(F(C-z_k)-s_k),\\
			r_k=&1+\sum_{i=1}^{k}\beta_i\phi_i^T\phi_i,
		\end{aligned}
	\right.
	\end{align}
	where $\alpha_k>0$ and $\beta_k>0$ are the step size coefficients.

	
\end{algorithm}
\begin{remark}\label{n1}
	The step size coefficients $\alpha_k$ and $\beta_k$ are designed to adjust the convergence rate. The case for $\alpha_k=\beta$ and $\beta_k=1$ is just the case in the recursive projection algorithm in \cite{guo2013recursive}. 	
\end{remark}


%

\newtheorem{lemma21}{Lemma}[section]

Next, we will establish the convergence and convergence rate of the RPFI algorithm. For the simplicity of description, denote 
\begin{align}\label{fandF}
	\begin{aligned}
		F_k=&F(C-\phi_k^T\theta), \hat{F}_k=F(C-z_k), \\
		f_k=&f(C-\phi_k^T\theta), \hat{f}_k=f(C-z_k).
	\end{aligned}
\end{align}

\begin{assumption}\label{A4}
	The coefficients $\alpha_k$ and $\beta_k$ satisfy\\ $0<\underline{\alpha}\leq\alpha_k\leq\bar{\alpha}<\infty$ and $0<\underline{\beta}\leq\beta_k\leq\bar{\beta}<\infty$.
\end{assumption}

\begin{theorem}\label{theorem1}
	If Assumptions \ref{A1}-\ref{A4} hold, the RPFI algorithm is both mean-square and almost surely convergent, i.e.,
	\begin{equation*}
		\lim\limits_{k\to\infty}E\Vert\widetilde{\theta}_k\Vert^2=0, \quad
		\lim\limits_{k\to\infty}\widetilde{\theta}_k=0, a.s.,
	\end{equation*}
	where $\widetilde{\theta}_k=\hat{\theta}_k-\theta$ is the estimation error.
\end{theorem}
\begin{proof}
	According to (\ref{eq1}), 	
		\begin{align*}
			\Vert\widetilde{\theta}_k\Vert^2=&[\widetilde{\theta}_{k-1}+\frac{1}{r_k}\alpha_k\gamma_k(\hat{F}_k-s_k)\phi_k]^T\\
			&\cdot[\widetilde{\theta}_{k-1}+\frac{1}{r_k}\alpha_k\gamma_k(\hat{F}_k-s_k)\phi_k]\\
			=&\widetilde{\theta}_{k-1}^T\widetilde{\theta}_{k-1}+2\frac{1}{r_k}\alpha_k\gamma_k(\hat{F}_k-s_k)\phi_k^T\widetilde{\theta}_{k-1}\\
			&+\frac{1}{r_k^2}\alpha_k^2\gamma_k^2(\hat{F}_k-s_k)^2\phi_k^T\phi_k.
		\end{align*}
	Denote $\mathcal{F}_k=\sigma\{d_s,s\leq k\}$. Noticing $E\left[s_k|\mathcal{F}_{k-1}\right]=EI_{\{y_k\leq C\}}=F_k$, according to the properties of conditional expectation in \cite{ash2014real} and Mean-Value theorem in \cite{Apostol:105425}, there exists $\xi_k$ between $C-z_k$ and $C-\phi_k^T\theta$ such that 
	\begin{align}\label{mean-value}
			&2E\left[\frac{1}{r_k}\alpha_k\gamma_k(\hat{F}_k-s_k)\phi_k^T\widetilde{\theta}_{k-1}\right]\nonumber\\
			=&2E\left[E\left[\frac{1}{r_k}\alpha_k\gamma_k(\hat{F}_k-s_k)\phi_k^T\widetilde{\theta}_{k-1}|\mathcal{F}_{k-1}\right]\right]\nonumber\\
			=&2E\left[\frac{1}{r_k}\alpha_k\gamma_k\phi_k^T\widetilde{\theta}_{k-1}(\hat{F}_k-F_k)\right]\nonumber\\
			=&2E\left[\frac{1}{r_k}\alpha_k\gamma_k\phi_k^T\widetilde{\theta}_{k-1}(-\widetilde{z}_kf(\xi_k))\right],
	\end{align}
	where $\widetilde{z}_k=z_k-\phi_k^T\theta$. 
	Denote $\grave{f}_k=f(\xi_k)$ and $ \underline{f}=\min\limits_{x\in\left[C-M,C+M\right]}f(x)$. By  $x\in\left[C-M,C+M\right]$ and Assumption \ref{A2}, we know $\underline{f}>0$. It follows that
	\begin{align}\label{eqconvergence}
		E\Vert\widetilde{\theta}_k\Vert^2
		=&E\Vert\widetilde{\theta}_{k-1}\Vert^2-2\frac{\alpha_k}{r_k}E[\gamma_k\widetilde{z}_k\grave{f}_k\phi_k^T\widetilde{\theta}_{k-1}]\nonumber\\
		&+\frac{\alpha_k^2}{r_k^2}\phi_k^T\phi_kE[\gamma_k^2(\hat{F}_k^2-2\hat{F}_ks_k+s_k^2)]\nonumber\\
		\leq&E\Vert\widetilde{\theta}_{k-1}\Vert^2-2\frac{\underline{\alpha}\underline{f}}{r_k}E\left[\gamma_k\widetilde{z}_k\phi_k^T\widetilde{\theta}_{k-1}\right]+2\frac{\bar{\alpha}^2}{r_k^2}\bar{\phi}^2E[\gamma_k^2]\nonumber\\
		\leq&E\Vert\widetilde{\theta}_{k-1}\Vert^2-2\frac{\underline{\alpha}\underline{f}}{r_k}E\left[\phi_k^T\widetilde{\theta}_{k-1}\widetilde{\theta}_{k-1}^T\phi_kI_{\{|\phi_k^T\hat{\theta}_{k-1}|\leq M\}}\right]\nonumber\\
		&-2\frac{\underline{\alpha}\underline{f}}{r_k}E\left[\phi_k^T\widetilde{\theta}_{k-1}|\phi_k^T\hat{\theta}_{k-1}|\widetilde{z}_kI_{\{|\phi_k^T\hat{\theta}_{k-1}|> M\}}\right]\nonumber\\
		&+2\frac{\bar{\alpha}^2}{r_k^2}\bar{\phi}^2E\left[\gamma_k^2\right].
	\end{align}	
    Since 
	$\phi_k^T\widetilde{\theta}_{k-1}$ and $\widetilde{z}_k$ have the same sign according to the definition of $z_k$ and $M$, we have
	\begin{align*}
		&E\left[\phi_k^T\widetilde{\theta}_{k-1}|\phi_k^T\hat{\theta}_{k-1}|\widetilde{z}_kI_{\{|\phi_k^T\hat{\theta}_{k-1}|> M\}}\right]\nonumber\\
		\geq&2E[|\phi_k^T\widetilde{\theta}_{k-1}|\cdot|\phi_k^T\hat{\theta}_{k-1}|I_{\{|\phi_k^T\hat{\theta}_{k-1}|>M\}}]\nonumber\\
		\geq&E[|\phi_k^T\widetilde{\theta}_{k-1}|\cdot|\phi_k^T\hat{\theta}_{k-1}|I_{\{|\phi_k^T\hat{\theta}_{k-1}|>M\}}]\nonumber\\
		&+E[|\phi_k^T\widetilde{\theta}_{k-1}| MI_{\{|\phi_k^T\hat{\theta}_{k-1}|>M\}}].
	\end{align*}
    Hence we have 
    \begin{align}\label{second}
    	&E\left[\phi_k^T\widetilde{\theta}_{k-1}\widetilde{\theta}_{k-1}^T\phi_kI_{\{|\phi_k^T\hat{\theta}_{k-1}|>M\}}\right]\nonumber\\
		\leq&E\left[|\phi_k^T\widetilde{\theta}_{k-1}|\cdot|\phi_k^T\hat{\theta}_{k-1}|I_{\{|\phi_k^T\hat{\theta}_{k-1}|>M\}}\right]\nonumber\\
        &+E\left[|\phi_k^T\widetilde{\theta}_{k-1}| MI_{\{|\phi_k^T\hat{\theta}_{k-1}|>M\}}\right]\nonumber\\
        \leq&E\left[\phi_k^T\widetilde{\theta}_{k-1}|\phi_k^T\hat{\theta}_{k-1}|\widetilde{z}_kI_{\{|\phi_k^T\hat{\theta}_{k-1}|> M\}}\right].
    \end{align}
	Moreover, by Assumptions \ref{A1} and \ref{A3}, we have 	
	\begin{align}\label{gamma}
		E\left[\gamma_k^2\right]\leq& 1+E\left[|\phi_k^T\hat{\theta}_{k-1}|^2\cdot I_{\{|\phi_k^T\hat{\theta}_{k-1}|>M\}}\right]\nonumber\\
		\leq&1+E\left[(2|\phi_k^T\widetilde{\theta}_{k-1}|^2+2|\phi_k^T\theta|^2)\cdot I_{\{|\phi_k^T\hat{\theta}_{k-1}|>M\}}\right]\nonumber\\
		\leq&1+2E\left[|\phi_k^T\widetilde{\theta}_{k-1}|^2\right]+2(\bar{\phi}\bar{\theta})^2.
	\end{align}
	Taking (\ref{second}) and (\ref{gamma}) into (\ref{eqconvergence}) yields
	\begin{align}\label{convergence speed}
		E\Vert\widetilde{\theta}_k\Vert^2
		\leq&E\Vert\widetilde{\theta}_{k-1}\Vert^2+\left(-2\frac{\underline{\alpha}\underline{f}}{r_k}+4\frac{\bar{\alpha}^2}{r_k^2}\bar{\phi}^2\right)E\left[(\phi_k^T\widetilde{\theta}_{k-1})^2\right]\nonumber\\
		&+2\frac{\bar{\alpha}^2}{r_k^2}\bar{\phi}^2\left[1+2(\bar{\phi}\bar{\theta})^2\right].
	\end{align}
	Assumption \ref{A3} illustrates that  $\frac{1}{k}tr\left(\sum_{i=1}^{k}\phi_i\phi_i^T\right)\geq n\delta^2$. Hence we have $r_k\geq1+\underline{\beta}\sum_{i=1}^{k}\phi_i^T\phi_i=1+\underline{\beta}tr\left(\sum_{i=1}^{k}\phi_i^T\phi_i\right)=1+\underline{\beta}tr\left(\sum_{i=1}^{k}\phi_i\phi_i^T\right)\geq1+k\underline{\beta}n\delta^2$. combining with $r_k\leq1+k\bar{\beta}\bar{\phi}^2$, we know that
	\begin{equation}\label{r_k}
		\frac{1}{r_k}=O\left(\frac{1}{k}\right).
	\end{equation}
	Hence there exists positive integer $i_0$, such that $-2\underline{\alpha}\underline{f}+\frac{4\bar{\alpha}^2\bar{\phi}^2}{r_{i+1}}<-\underline{\alpha}\underline{f}$, for $i\geq i_0$. So when $k> i_0$, we have
	\begin{align}\label{bounded_Etheta}
		E\Vert\widetilde{\theta}_k\Vert^2
		\leq&E\Vert\widetilde{\theta}_{i_0}\Vert^2-\underline{\alpha}\underline{f}\sum_{i=i_0}^{k-1}\frac{1}{r_{i+1}}E|\phi_{i+1}^T\widetilde{\theta}_i|^2\nonumber\\
		&+\sum_{i=i_0}^{k-1}\frac{2\bar{\alpha}^2\bar{\phi}^2}{r_{i+1}^2}\left[1+2(\bar{\phi}\bar{\theta})^2\right].
	\end{align} 
	By (\ref{r_k}), we have $\sum_{i=i_0}^{k}\frac{1}{r_{i+1}}=\infty$ and $\sum_{i=i_0}^{k}\frac{2\bar{\alpha}^2\bar{\phi}^2}{r_{i+1}^2}<\infty$. Hence, the second term on RHS of (\ref{bounded_Etheta}) satisfies
	\begin{align}\label{phitheta}
		\underline{\alpha}\underline{f}\sum_{i=i_0}^{k-1}\frac{1}{r_{i+1}}E|\phi_{i+1}^T\widetilde{\theta}_i|^2<\infty,
	\end{align} otherwise it contradicts with $E\Vert\widetilde{\theta}_k\Vert^2\geq0$. Combining with $\sum_{i=i_0}^{k}\frac{1}{r_{i+1}}=\infty$, (\ref{phitheta}) further implies $\lim\limits_{k\to\infty}E|\phi_{k+1}^T\widetilde{\theta}_k|^2=0$, which combining with (\ref{gamma}) yields that there exists $D_1$ such that 
	\begin{align}\label{gammaD_1}
		E\left[\gamma_k^2\right]<D_1 \mbox{, for all k}.
	\end{align}
	Hence the RHS of (\ref{bounded_Etheta}) is bounded, which means there exists $N_1$ such that 
	\begin{align}\label{EN_1}
		E\Vert\widetilde{\theta}_k\Vert^2<N_1 \mbox{, for all k}.
	\end{align}
	Noticing that
	\begin{align}\label{thetajiantheta}
		\Vert\widetilde{\theta}_{k+j}-\widetilde{\theta}_k\Vert=&\Vert\sum_{i=k+1}^{k+j}(\hat{\theta}_i-\hat{\theta}_{i-1})\Vert\leq\sum_{i=k+1}^{k+j}\Vert\hat{\theta}_i-\hat{\theta}_{i-1}\Vert\nonumber\\
		=&\sum_{i=k+1}^{k+j}\Vert\frac{\phi_i\gamma_i\alpha_i}{r_i}(F(C-z_i)-s_i)\Vert\nonumber\\
		\leq&\frac{\bar{\phi}\bar{\alpha}}{r_{k+1}}\sum_{i=k+1}^{k+j}\Vert\gamma_i\Vert.
	\end{align} 
	By Cr-inequality in \cite{HanFu1991ProbabilityTP}, we have  $E\left(\sum_{i=k}^{k+j}\Vert\gamma_i\Vert\right)^2\leq j\sum_{i=k}^{k+j}E\Vert\gamma_i\Vert^2$ for $j=0,1,...,N-2.$ Hence, it can be seen 
	\begin{align*}
		&E\left(\phi_{k+j+1}^T\widetilde{\theta}_{k-1}\right)^2\\
		\leq&E\left(\phi_{k+j+1}^T\widetilde{\theta}_{k+j}\right)^2+\bar{\phi}^2E\left(\Vert\widetilde{\theta}_{k-1}-\widetilde{\theta}_{k+j}\Vert^2\right)\\
		&+2\bar{\phi}^2E\left(\Vert\widetilde{\theta}_{k+j}\Vert\cdot\Vert\widetilde{\theta}_{k-1}-\widetilde{\theta}_{k+j}\Vert\right)\nonumber\\
		\leq&E\left(\phi_{k+j+1}^T\widetilde{\theta}_{k+j}\right)^2+\frac{j\bar{\phi}^4\bar{\alpha}^2}{r_k^2}\sum_{i=k}^{k+j}E\Vert\gamma_i\Vert^2\\
		&+\frac{2\sqrt{j}\bar{\phi}^3\bar{\alpha}}{r_k}\sqrt{N_1\sum_{i=k}^{k+j}E\Vert\gamma_i\Vert^2}\nonumber\\
		\leq&E\left(\phi_{k+j+1}^T\widetilde{\theta}_{k+j}\right)^2+\frac{j\bar{\phi}^4\bar{\alpha}^2}{r_k^2}(j+1)D_1\\
		&+\frac{2\sqrt{j}\bar{\phi}^3\bar{\alpha}}{r_k}\sqrt{N_1(j+1)D_1}\nonumber\\
		\to&0, \text{ as } k\to\infty.
	\end{align*}
	Hence, by Assumption \ref{A3}, we get  
	\begin{align*}
		\delta E\Vert\widetilde{\theta}_{k-1}\Vert^2&\leq E\widetilde{\theta}_{k-1}^T\left(\sum_{j=0}^{N-1}\phi_{k+j+1}\phi_{k+j+1}^T\right)\widetilde{\theta}_{k-1}^T\\
		&=\sum_{j=0}^{N-1}E\left(\phi_{k+j+1}^T\widetilde{\theta}_{k-1}\right)^2\to0, \text{ as } k\to\infty,
	\end{align*}
	which implies $\lim\limits_{k\to\infty}E\Vert\widetilde{\theta}_k\Vert^2=0.$

	On the other hand, we have $$E\left[\Vert\widetilde{\theta}_k\Vert^2|\mathcal{F}_{k-1}\right]\leq\Vert\widetilde{\theta}_{k-1}\Vert^2+\frac{2\bar{\alpha}^2\bar{\phi}^2\gamma_k^2}{r_k^2},$$ and $$E\left[\sum_{i=1}^{k}\frac{2\bar{\alpha}^2\bar{\phi}^2\gamma_i^2}{r_i^2}\right]=\sum_{i=1}^{k}\frac{2\bar{\alpha}^2\bar{\phi}^2}{r_i^2}E[\gamma_i^2]<\infty.$$By Lemma \ref{lemmaconvergence}, we know that $\widetilde{\theta}_k$ converges almost surely to a finite limit. 
	Combining with $\lim\limits_{k\to\infty}E\Vert\widetilde{\theta}_k\Vert^2=0$ illustrates that $\widetilde{\theta}_k$ almost surely converges to 0. 
\end{proof}
\begin{remark}
	Now it can be explained clearly why we need to introduce the cut-off coefficient $z_k$ and the adaptive accelerated coefficient $\gamma_k$. In summary, $z_k$ plays a role in proving the mean square convergence and $\gamma_k$ is aimed at accelerating the convergence speed. In the analysis of Equation (\ref{eqconvergence}), we aim to obtain the contraction of $E\Vert\widetilde{\theta}_k\Vert^2$ with respect to $E\Vert\widetilde{\theta}_{k-1}\Vert^2$ so that the mean square convergence can be proved later. The contraction requires the second term on RHS of (\ref{eqconvergence}) to be negative. Therefore, we need $f(\xi_k)>0$. In \cite{guo2013recursive,8623153,guo2014identification,zhang2019asymptotically,wang2022unified,10247592,wang2023asymptotically}, this was obtained by using projection operators to ensure the estimates are bounded. Actually, the projection operator is unnecessary. We can get $f(\xi_k)>0$ as long as $\phi_k^T\hat{\theta}_{k-1}$ is bounded. Therefore, the cut-off coefficient $z_k$ is introduced. However, it can be seen from (\ref{eqconvergence}) that since $z_k$ is bounded, $\widetilde{z}_k$ is bounded as well. When the estimate $\hat{\theta}_{k-1}$ deviates significantly from the true parameter $\theta$, $\widetilde{z}_k$ is much smaller than  $\phi_k^T\widetilde{\theta}_{k-1}$. This means the contractive effect of the second term on RHS of (\ref{eqconvergence}) is not good enough, which can not reach $O\left(\Vert\widetilde{\theta}_{k-1}\Vert^2\right)$. Hence we designed the adaptive accelerated coefficient $\gamma_k$. Finally we would like to mention that compared with the high-dimensional projection operators, both $z_k$ and $\gamma_k$ are one-dimensional scalar comparisons, which means the computational complexity is reduced.

\end{remark}

\begin{theorem}
	If Assumptions \ref{A1}-\ref{A4} hold and the coefficients $\alpha_k$ and $\beta_k$ satisfy  $\frac{2\underline{\alpha}\underline{f}\delta^2}{1+\bar{\beta}\bar{\phi}^2}>1$, the estimation error of RPFI algorithm has the following property:
	\begin{align*}
		E\Vert\widetilde{\theta}_k\Vert^2=O\left(\frac{1}{k}\right).
	\end{align*}
\end{theorem}
\begin{proof}
	Noticing that 	
	\begin{align}\label{phitheta_high}
			&(\phi_i^T\widetilde{\theta}_{i-1})^2\nonumber\\
			=&[\phi_i^T(\widetilde{\theta}_{i-1}-\widetilde{\theta}_{k-N})]^2+(\phi_i^T\widetilde{\theta}_{k-N})^2\nonumber\\
			&+2\phi_i^T(\widetilde{\theta}_{i-1}-\widetilde{\theta}_{k-N})\phi_i^T\widetilde{\theta}_{k-N}\nonumber\\
			\geq&\widetilde{\theta}_{k-N}^T\phi_i\phi_i^T\widetilde{\theta}_{k-N} +2(\widetilde{\theta}_{i-1}-\widetilde{\theta}_{k-N})^T\phi_i\phi_i^T\widetilde{\theta}_{k-N}.\displaybreak
	\end{align}
	By Assumption \ref{A3}, (\ref{eq1}) and (\ref{phitheta_high}), we have	
	\begin{align}\label{high-rank}
		&-\sum_{i=k-N+1}^{k}\frac{\widetilde{\theta}_{i-1}^T\phi_i\phi_i^T\widetilde{\theta}_{i-1}}{r_i}\nonumber\\
		\leq&-\sum_{i=k-N+1}^{k}\frac{\widetilde{\theta}_{k-N}^T\phi_i\phi_i^T\widetilde{\theta}_{k-N} }{r_i}\nonumber\\
		&-\sum_{i=k-N+1}^{k}\frac{2(\widetilde{\theta}_{i-1}-\widetilde{\theta}_{k-N})^T\phi_i\phi_i^T\widetilde{\theta}_{k-N}}{r_i}\nonumber\\
		\leq&-\frac{1}{r_k}\widetilde{\theta}_{k-N}^T\left(\sum_{i=k-N+1}^{k}\phi_i\phi_i^T\right)\widetilde{\theta}_{k-N}\nonumber\\
		&-\sum_{i=k-N+1}^{k}\frac{2(\widetilde{\theta}_{i-1}-\widetilde{\theta}_{k-N})^T\phi_i\phi_i^T\widetilde{\theta}_{k-N}}{r_i}\nonumber\\
		\leq&-\frac{N\delta^2\Vert\widetilde{\theta}_{k-N}\Vert^2}{r_k}-\sum_{i=k-N+1}^{k}\frac{2(\widetilde{\theta}_{i-1}-\widetilde{\theta}_{k-N})^T\phi_i\phi_i^T\widetilde{\theta}_{k-N}}{r_i}.
	\end{align}
	Substituting (\ref{high-rank}) into the first inequality of (\ref{convergence speed}) yields
	\begin{align}\label{equation4}
		E\Vert\widetilde{\theta}_{k}\Vert^2
		\leq&E\Vert\widetilde{\theta}_{k-N}\Vert^2-2\underline{\alpha}\underline{f}\sum_{i=k-N+1}^{k}E\frac{\widetilde{\theta}_{i-1}^T\phi_i\phi_i^T\widetilde{\theta}_{i-1}}{r_i}\nonumber\\
		&+\sum_{i=k-N+1}^{k}\frac{2\bar{\alpha}^2\bar{\phi}^2}{r_i^2}\left[\gamma_i^2\right]\nonumber\\         
		\leq&E\Vert\widetilde{\theta}_{k-N}\Vert^2-\frac{2\underline{\alpha}\underline{f}N\delta^2}{r_k}E\Vert\widetilde{\theta}_{k-N}\Vert^2\nonumber\\
		&-\sum_{i=k-N+1}^{k}4\underline{\alpha}\underline{f}E\frac{(\widetilde{\theta}_{i-1}-\widetilde{\theta}_{k-N})^T\phi_i\phi_i^T\widetilde{\theta}_{k-N}}{r_i}\nonumber\\
		&+\sum_{i=k-N+1}^{k}\frac{2\bar{\alpha}^2\bar{\phi}^2}{r_{i+1}^2}\left[\gamma_k^2\right].
	\end{align}
	We now analyze the third term and fourth term on RHS of (\ref{equation4}). 
	By (\ref{r_k}) and (\ref{gammaD_1})-(\ref{thetajiantheta}), the third term on RHS of (\ref{equation4}) can be written as
	\begin{align}\label{equation5}
		&-\sum_{i=k-N+1}^{k}4\underline{\alpha}\underline{f}E\frac{(\widetilde{\theta}_{i-1}-\widetilde{\theta}_{k-N})^T\phi_i\phi_i^T\widetilde{\theta}_{k-N}}{r_i}\nonumber\\
		\leq&\frac{4\bar{\phi}^2\underline{\alpha}\underline{f}}{r_{k-N+1}}\sum_{i=k-N+1}^{k}E(\Vert\widetilde{\theta}_{i-1}-\widetilde{\theta}_{k-N}\Vert\cdot\Vert\widetilde{\theta}_{k-N}\Vert)\nonumber\\
		\leq&\frac{4\bar{\phi}^3\underline{\alpha}\underline{f}\bar{\alpha}}{r_{k-N+1}^2}\sum_{i=k-N+1}^{k}\sum_{j=k-N+1}^{i}\sqrt{E\Vert\gamma_j\Vert^2}\cdot\sqrt{E\Vert\widetilde{\theta}_{k-N}\Vert^2}\nonumber\\
		=&O\left(\frac{1}{(k-N+1)^2}\right).
	\end{align}	
	By (\ref{r_k}) and (\ref{gammaD_1}), the fourth term on RHS of (\ref{equation4}) can be written as
	\begin{align}\label{equation6}
		\sum_{i=k-N+1}^{k}\frac{2\bar{\alpha}^2\bar{\phi}^2}{r_i^2}E\left[\gamma_i\right]
		=O\left(\frac{1}{(k-N+1)^2}\right).
	\end{align}
%
	From (\ref{equation4})-(\ref{equation6}), we get
	\begin{align}\label{equation7}
		&E\Vert\widetilde{\theta}_k\Vert^2\nonumber\\
		\leq& E\Vert\widetilde{\theta}_{k-N}\Vert^2-\frac{2\underline{\alpha}\underline{f}N\delta^2}{r_k}E\Vert\widetilde{\theta}_{k-N}\Vert^2\nonumber\\
		&+O\left(\frac{1}{(k-N+1)^2}\right)\nonumber\\
		\leq&\prod_{j=0}^{ \left\lfloor \frac{k}{N}\right\rfloor-1}\left(1-\frac{2\underline{\alpha}\underline{f} N\delta^2}{r_{k-jN}}\right)E\Vert\widetilde{\theta}_{k- \left\lfloor \frac{k}{N}\right\rfloor N}\Vert^2\nonumber\\
		&+\sum_{j=1}^{ \left\lfloor \frac{k}{N}\right\rfloor}\prod_{i=0}^{j-1}\left(1-\frac{2\underline{\alpha}\underline{f} N\delta^2}{r_{k-iN}}\right)O\left(\frac{1}{(k-jN+1)^2}\right),
	\end{align}
    where $\left\lfloor x \right\rfloor=\max\{a\in \mathbb{Z}|a\leq x\},\left \lceil x \right \rceil =\min\{a\in \mathbb{Z}|a\geq x\},$ and $\kappa=\left\lceil \frac{k}{N}\right\rceil - \left\lfloor \frac{k}{N}\right\rfloor$. 
	For the first term on RHS of (\ref{equation7}), we get
	\begin{align}\label{equation8}
		\prod_{j=0}^{ \left\lfloor \frac{k}{N}\right\rfloor-1}\left(1-\frac{2\underline{\alpha}\underline{f} N\delta^2}{r_{k-jN}}\right)\leq&\prod_{m=\kappa+1}^{\left\lceil \frac{k}{N}\right\rceil}\left(1-\frac{2\underline{\alpha}\underline{f} N\delta^2}{r_{mN}}\right)\nonumber\\
		\leq&\prod_{m=\kappa+1}^{\left\lceil \frac{k}{N}\right\rceil}\left(1-\frac{2\underline{\alpha}\underline{f} \delta^2}{1+m\bar{\beta}\bar{\phi}^2}\right)\nonumber\\
		=&O\left(\left(\frac{1}{k^{\frac{2\underline{\alpha}\underline{f}\delta^2}{1+\bar{\beta}\bar{\phi}^2}}}\right)\right).
	\end{align}	
	Similarly to (\ref{equation8}), according to (\ref{r_k}) and Lemma \ref{lemmabound}, the second term on RHS of (\ref{equation7}) can be written as 	
	\begin{align}\label{equation9}
		&\sum_{j=1}^{ \left\lfloor \frac{k}{N}\right\rfloor}\prod_{i=0}^{j-1}\left(1-\frac{2\underline{\alpha}\underline{f} N\delta^2}{r_{k-iN}}\right)\cdot O\left(\frac{1}{(k-jN+1)^2}\right)\nonumber\\
		=&O\left(\prod_{j=0}^{ \left\lfloor \frac{k}{N}\right\rfloor-1}(1-\frac{2\underline{\alpha}\underline{f} N\delta^2}{r_{k+1-jN}})\right)\nonumber\\
		&+\sum_{m=1}^{\left\lfloor \frac{k}{N}\right\rfloor-1}\prod_{p=\kappa+m+1}^{\left\lceil \frac{k}{N}\right\rceil}\left(1-\frac{2\underline{\alpha}\underline{f} N\delta^2}{r_{pN}}\right)O\left(\frac{1}{m^2N^2}\right)\nonumber\\
		=&O\left(\prod_{m=\kappa+1}^{\left\lceil \frac{k}{N}\right\rceil}(1-\frac{\frac{2\underline{\alpha}\underline{f}\delta^2}{1+\bar{\beta}\bar{\phi}^2}}{m})\right)\nonumber\displaybreak\\
		&+\sum_{m=1}^{\left\lfloor \frac{k}{N}\right\rfloor-1}\prod_{p=\kappa+m+1}^{\left\lceil \frac{k}{N}\right\rceil}\left(1-\frac{\frac{2\underline{\alpha}\underline{f}\delta^2}{1+\bar{\beta}\bar{\phi}^2}}{p}\right)O\left(\frac{1}{m^2N^2}\right)\nonumber\\
		=&\left\{
		\begin{aligned}
			&O\left(\left(\frac{1}{k^{\frac{2\underline{\alpha}\underline{f}\delta^2}{1+\bar{\beta}\bar{\phi}^2}}}\right)\right), & \frac{2\underline{\alpha}\underline{f}\delta^2}{1+\bar{\beta}\bar{\phi}^2}&<1;\\
			&O\left(\frac{\log k}{k}\right), & \frac{2\underline{\alpha}\underline{f}\delta^2}{1+\bar{\beta}\bar{\phi}^2}&=1;\\
			&O\left(\frac{1}{k}\right), & \frac{2\underline{\alpha}\underline{f}\delta^2}{1+\bar{\beta}\bar{\phi}^2}&>1.
		\end{aligned}
		\right.
	\end{align}

	Taking (\ref{equation8}) and (\ref{equation9}) into (\ref{equation7}) yields $E\Vert\widetilde{\theta}_k\Vert^2=O\left(\frac{1}{k}\right)$, if $\frac{2\underline{\alpha}\underline{f}\delta^2}{1+\bar{\beta}\bar{\phi}^2}>1$.
\end{proof}

\section{Asymptotically efficient algorithm}\label{asymptotic efficiency}

The previous section has established that the RPFI algorithm converges in both mean square and almost surely. And the mean square convergence rate is proved to be $O\left(\frac{1}{k}\right)$. This section will discuss how to construct the optimal identification algorithm with binary-valued observations under low computational complexity.

We use the CR lower bound rather than the least square method to construct the optimal identification algorithm. Here are the reasons. On the one hand, due to the strong nonlinearity in binary-valued systems, their likelihood functions lack explicit solutions. It makes the identification algorithms unable to be designed similar to least squares algorithms, which minimizes the objective function as the optimal criterion of designing identification algorithms. On the other hand, we notice that the CR bound is the lower bound of the unbiased estimate covariance. Therefore, we generally use the CR lower bound as the optimality criterion of binary-valued identification. Therefore, inspired by \cite{zhang2019asymptotically,wang2023asymptotically}, we design 
\begin{align*}
	\hat{\alpha}_k=\frac{\hat{f}_k}{\hat{F}_k(1-\hat{F}_k)},\quad
	\hat{\beta}_k=\frac{\hat{f}_k^2}{\hat{F}_k(1-\hat{F}_k)},
\end{align*}
to make adaptive step size. And we use
\begin{align*}
	\hat{P}_k=\hat{P}_{k-1}-\frac{\hat{\beta}_k\hat{P}_{k-1}\phi_k\phi_k^T\hat{P}_{k-1}^T}{1+\hat{\beta}_k\phi_k^T\hat{P}_{k-1}\phi_k},
\end{align*} 
to approximate the CR lower bound. But unlike \cite{zhang2019asymptotically,wang2023asymptotically}, this paper removes the projection operator to reduce the computational complexity. To make full use of the prior information, we use a cut-off coefficient $z_k$ to remove the projection operators and introduce an adaptive accelerated coefficient $\gamma_k$ to maintain high convergence rate similarly to the RPFI algorithm.

Based on above, we give the IMPF algorithm as follows:
\begin{algorithm}[h]
	\caption{The IMPF Algorithm}
	\label{Algorithm2}
	For the  arbitrary initial value $\hat{\theta}_0\in\mathbb{R}^n$ and an positive definitive matrix $\hat{P}_0\in\mathbb{R}^{n\times n}$, the algorithm is recursively defined at any $k\geq0$ as follows:
	
	Step 1: Update of the cut-off coefficient and the adaptive accelerated coefficient:
	\begin{align}\label{eq2.1}
		\begin{aligned}
			z_k=&
			\left\{
			\begin{aligned}
				&M, & \phi_k^T\hat{\theta}_{k-1}&>M;\\
				&\phi_k^T\hat{\theta}_{k-1}, & \phi_k^T\hat{\theta}_{k-1}&\in[-M,M];\\
				&-M, & \phi_k^T\hat{\theta}_{k-1}&<-M,
			\end{aligned}
			\right.\\
			\gamma_k=&
			\left\{
			\begin{aligned}
				&1, & |\phi_k^T\hat{\theta}_{k-1}|&\leq M;\\
				&|\phi_k^T\hat{\theta}_{k-1}|, & |\phi_k^T\hat{\theta}_{k-1}|&>M,
			\end{aligned}
			\right.
		\end{aligned}
	\end{align}
	where $M=\bar{\phi}\bar{\theta}+2$.
	
	Step 2: Update of the adaptive step coefficients:
	\begin{align}
		\begin{aligned}
			\hat{\alpha}_k=&\frac{\hat{f}_k}{\hat{F}_k(1-\hat{F}_k)},\quad
			\hat{\beta}_k=&\frac{\hat{f}_k^2}{\hat{F}_k(1-\hat{F}_k)},
		\end{aligned}
	\end{align}
	where $\hat{F}_k=F(C-z_k)$, $\hat{f}_k=f(C-z_k)$.
	
	Step 3: Estimation:
	\begin{align}\label{eq2}
		\left\{
		\begin{aligned}
			\hat{\theta}_{k}=&\hat{\theta}_{k-1}+\hat{P}_k\phi_k\gamma_k\hat{\alpha}_k(\hat{F}_k-s_k),\\
			\hat{P}_k=&\hat{P}_{k-1}-\frac{\hat{\beta}_k\hat{P}_{k-1}\phi_k\phi_k^T\hat{P}_{k-1}^T}{1+\hat{\beta}_k\phi_k^T\hat{P}_{k-1}\phi_k}.
		\end{aligned}
	\right.
	\end{align}

\end{algorithm}
\begin{remark}
	The IMPF algorithm appears somewhat similar to the least square method in its structure. But due to binary-valued observations, the process of deriving the IMPF algorithm and the analysis of its properties differ fundamentally from the least square method. Compared to the algorithm in \cite{Ke2022RecursiveIO}, the IMPF algorithm makes full use of the prior information and designs adaptive weight coefficients to approximate the CR lower bound. Since the projection operator is removed, the computational complexity of this algorithm is much lower than other asymptotically efficient algorithm in \cite{zhang2019asymptotically, wang2023asymptotically}. 
\end{remark}

\begin{remark}\label{r}
	Compared with the RPFI algorithm, the IMPF algorithm has adaptive weight coefficients and it uses $\hat{P}_k$ rather than $r_k$ since $\hat{P}_k$ can represent the covariance of estimate error. $P_k$ makes it difficult to analyze the performance of the IMPF algorithm for high-order FIR systems. So next we will give the theoretical analysis of the IMPF algorithm for first-order FIR systems.
\end{remark}

\begin{theorem}\label{1k}
	For the first-order FIR system (\ref{M}) with binary-valued observations (\ref{moxingtwo}), if Assumptions \ref{A1}-\ref{A3} hold, the IMPF algorithm is convergent in mean square with 
	\begin{align*}
		E\widetilde{\theta}_k^2=O\left(\frac{1}{k}\right).
	\end{align*}
\end{theorem}
\begin{proof}
	We divide the proof into two steps.
	
	\textbf{Step 1:}
	We will prove that there exists $\mu\in\mathbb{R}^+$ satisfying 
	\begin{equation}\label{step1}
		\begin{aligned}
			E\widetilde{\theta}_k^{2m}=
			\left\{
			\begin{aligned}
				&O\left(\frac{1}{k^{2m\mu}}\right), &2\mu&<1,\\
				&O\left(\frac{\log^m k}{k^m}\right), &2\mu&=1,\\
				&O\left(\frac{1}{k^m}\right), &2\mu&>1.\\
			\end{aligned}
			\right.\\
		\end{aligned}
	\end{equation}
    According to Assumption \ref{A2} and the boundedness of $z_k$, we know that $f_k, F_k, \hat{f}_k, \hat{F}_k$ defined in (\ref{fandF}) are bounded. 
	Denote 
	\begin{align}\label{Pbeta}
		\begin{aligned}
			\alpha_k&=\frac{f_k}{F_k(1-F_k)},\quad \beta_k=\frac{f_k^2}{F_k(1-F_k)},\\
			P_k&=P_{k-1}-\frac{\beta_kP_{k-1}\phi_k\phi_k^TP_{k-1}^T}{1+\beta_k\phi_k^TP_{k-1}\phi_k},
		\end{aligned}
	\end{align}
	and we know that there exist $\underline{\alpha}$, $\bar{\alpha}$, $\underline{\beta}$ and $\bar{\beta}$ such that $0<\underline{\alpha}\leq\alpha_k, \hat{\alpha}_k\leq\bar{\alpha}<\infty$ and $0<\underline{\beta}\leq\beta_k, \hat{\beta}_k\leq\bar{\beta}<\infty$. 
	By (\ref{eq2}) and (\ref{Pbeta}), we have $\hat{P}_k^{-1}=\sum_{i=1}^{k}\hat{\beta}_i\phi_i^2+\hat{P}_0^{-1}$ and $P_k^{-1}=\sum_{i=1}^{k}\beta_i\phi_i^2+P_0^{-1}$. 
	Similarly to the proof of (\ref{r_k}), it can be seen that
	\begin{align}\label{P_k}
		\hat{P}_k=O\left(\frac{1}{k}\right),\quad P_k=O\left(\frac{1}{k}\right).
	\end{align}
	According (\ref{eq2}), we have
	\begin{equation*}
		\begin{aligned}
			\widetilde{\theta}_k^2
			=\widetilde{\theta}_{k-1}^2+2\hat{P}_k\hat{\alpha}_k\gamma_k(\hat{F}_k-s_k)\phi_k\widetilde{\theta}_{k-1}+\hat{P}_k^2\hat{\alpha}_k^2\phi_k^2\gamma_k^2(\hat{F}_k-s_k)^2.
		\end{aligned}
	\end{equation*}
    Similarly to the proof of (\ref{mean-value}), there exists $\zeta_k$ between $C-z_k$ and $C-\phi_k\theta$ such that 
	\begin{equation*}
		\begin{aligned}
			&2E\left[\hat{P}_k\hat{\alpha}_k\gamma_k(\hat{F}_k-s_k)\phi_k\widetilde{\theta}_{k-1}\right]\\
			=&2E\left[\hat{P}_k\hat{\alpha}_k\phi_k\gamma_k\widetilde{\theta}_{k-1}(-\widetilde{z}_kf(\zeta_k))\right],
		\end{aligned}
	\end{equation*}
	where $\widetilde{z}_k=z_k-\phi_k\theta$. Noticing $|\hat{F}_k|\leq1$ and  $|s_k|\leq1$, it can be derived that
	\begin{align}\label{shelambda}
		E\widetilde{\theta}_k^2
		\leq& E\widetilde{\theta}_{k-1}^2+2E\left[\hat{P}_k\hat{\alpha}_k\phi_k\gamma_k\widetilde{\theta}_{k-1}(-\widetilde{z}_kf(\zeta_k))\right]\nonumber\\
		&+4E\left[\hat{P}_k^2\hat{\alpha}_k^2\phi_k^2\gamma_k^2\right].
	\end{align}
    For the second term on RHS of (\ref{shelambda}), let $\check{f}_k=f(\zeta_k), \underline{f}=\min\limits_{x\in\left[C-M,C+M\right]}f(x),  \lambda_k=\frac{\hat{P}_k\hat{\alpha}_k}{P_k\beta_k}\check{f}_k$. According to the boundedness of $z_{k-1},\theta,\xi_{k-1}$, and the continuity of $\lambda_k$, we learn  $\underline{\lambda}=\inf\limits_{k}\min\limits_{z_{k-1},\xi_{k-1}}\lambda_k>0$  and $\bar{\lambda}=\sup\limits_{k}\max\limits_{z_{k-1},\xi_{k-1}}\lambda_k<\infty.$ 
    Substituting $\lambda_k$ into (\ref{shelambda}), we have 
	\begin{align}\label{Beq1}
		E\widetilde{\theta}_k^2
		\leq&E\widetilde{\theta}_{k-1}^2-2P_k\beta_{k}E\left[\phi_k^2\lambda_k\widetilde{\theta}_{k-1}^2I_{\{|\phi_k\hat{\theta}_{k-1}|\leq M\}}\right]\nonumber\\
		&-2P_k\beta_{k}E\left[\phi_k\lambda_k\widetilde{\theta}_{k-1}|\phi_k\hat{\theta}_{k-1}|\widetilde{z}_kI_{\{|\phi_k\hat{\theta}_{k-1}|> M\}}\right]\nonumber\\
		&+4\frac{P_k^2\beta_k^2\bar{\lambda}^2}{\underline{f}^2}\phi_k^2E\left[\gamma_k^2\right].
	\end{align}
	Similarly to the proof of (\ref{second}), we can prove that 
	\begin{align}\label{dierxianggaoci}
		&E[\widetilde{\theta}_{k-1}^{2m-1}|\phi_k\hat{\theta}_{k-1}|\phi_k\lambda_k\widetilde{z}_kI_{\{|\phi_k\hat{\theta}_{k-1}|>M\}}]\nonumber\\
		\geq&E[\widetilde{\theta}_{k-1}^{2m}\phi_k^2\lambda_kI_{\{|\phi_k\hat{\theta}_{k-1}|>M\}}], k\in \mathbb{Z}^+,m\in\mathbb{Z}^+.
	\end{align}  
	In order to deal with the fourth term of (\ref{Beq1}), we will prove that 
	for all $t\geq1$, $E\left[\gamma_k^t\right]$ is bounded. 
		Similarly to (\ref{gamma}), we have
		\begin{align}\label{gamma_high}
			\gamma_k^t
			\leq&1+2^{t-1}|\phi_k\widetilde{\theta}_{k-1}|^t+2^{t-1}(\bar{\phi}\bar{\theta})^t,
		\end{align}   
	and
		\begin{align}\label{E_gamma_high}
			E\left[\gamma_k^t\right]
			\leq&1+2^{t-1}\bar{\phi}^tE\left[|\widetilde{\theta}_{k-1}|^t\right]+2^{t-1}(\bar{\phi}\bar{\theta})^t.
		\end{align}
		Hence we transform the problem of proving $E\left[\gamma_k^t\right]<\infty$ into the problem of proving $E\left[|\widetilde{\theta}_k|^t\right]<\infty$. Next, we will prove that $E\widetilde{\theta}_k^{2t}$ is bounded for all $t\geq1$.
		If $t=1$,
		\begin{equation*}
			\begin{aligned}
				E\widetilde{\theta}_k^2
				\leq&\left(1-2\underline{\lambda}P_k\beta_{k}\phi_k^2+8\frac{P_k^2\beta_{k}^2\phi_k^4\bar{\lambda}^2}{\underline{f}^2}\right)E\widetilde{\theta}_{k-1}^2\\
				&+4\frac{P_k^2\beta_{k}^2\phi_k^2\bar{\lambda}^2}{\underline{f}^2}\left(1+2(\bar{\phi}\bar{\theta})^2\right).
			\end{aligned}
		\end{equation*}
		Since $P_k=O\left(\frac{1}{k}\right)$, there exists $i_0$ such that when  $i>i_0$, $1-2\underline{\lambda}P_k\beta_{k}\phi_k^2+8\frac{P_k^2\beta_{k}^2\phi_k^4\bar{\lambda}^2}{\underline{f}^2}<1$. So we get
		\begin{align*}
				E\widetilde{\theta}_k^2\leq&\prod_{i=i_0+1}^{k}\left(1-2\underline{\lambda}P_i\beta_{i}\phi_i^2+8\frac{P_i^2\beta_{i}^2\phi_i^4\bar{\lambda}^2}{\underline{f}^2}\right)E\widetilde{\theta}_{i_0}^2\\
				&+\sum_{l=i_0+1}^{k-1}\prod_{i=l+1}^{k}\left(1-2\underline{\lambda}P_i\beta_{i}\phi_i^2+8\frac{P_i^2\beta_{i}^2\phi_i^4\bar{\lambda}^2}{\underline{f}^2}\right)\\
				&\cdot\left(\frac{2P_{l-1}^2\beta_{l}^2\phi_l^2\bar{\lambda}^2}{\underline{f}^2}\right)\left(1+2(\bar{\phi}\bar{\theta})^2\right)\displaybreak\\
				\leq&E\widetilde{\theta}_{i_0}^2+\sum_{l=i_0+1}^{k-1}O\left(\frac{1}{l^2}\right)<\infty.
		\end{align*}
		Hence $E\widetilde{\theta}_k^2<\infty$ for all $k\geq1$. When $t=2$, according to the algorithm (\ref{eq2}) we have 
		\begin{align}\label{high_E}
			E\widetilde{\theta}_k^4=&E\widetilde{\theta}_{k-1}^4+4E\left[\widetilde{\theta}_{k-1}^3\hat{P}_k\hat{\alpha}_k\phi_k\gamma_k(\hat{F}_k-s_k)\right]\nonumber\\
			&+6E\left[\widetilde{\theta}_{k-1}^2\hat{P}_k^2\hat{\alpha}_k^2\phi_k^2\gamma_k^2(\hat{F}_k-s_k)^2\right]\nonumber\\
			&+4E\left[\widetilde{\theta}_{k-1}\hat{P}_k^3\hat{\alpha}_k^3\phi_k^3\gamma_k^3(\hat{F}_k-s_k)^3\right]\nonumber\\
			&+E\left[\hat{P}_k^4\hat{\alpha}_k^4\phi_k^4\gamma_k^4(\hat{F}_k-s_k)^4\right].
		\end{align}
		Then it can be seen that
		\begin{align}\label{Ex4_more}
			E\widetilde{\theta}_k^4\leq&\left(1-4\underline{\lambda}P_k\beta_{k}\phi_k^2+\frac{48P_k^2\beta_k^2\phi_k^4\bar{\lambda}^2}{\underline{f}^2}+\frac{128P_k^3\beta_k^3\phi_k^6\bar{\lambda}^3}{\underline{f}^3}\notag\right.\\
			&\left.+\frac{128P_k^4\beta_k^4\phi_k^8\bar{\lambda}^4}{\underline{f}^4}\right)E\widetilde{\theta}_{k-1}^4+O\left(\frac{1}{k^2}\right).
		\end{align}
		There exists $i_1$ such that when $i>i_1$, $1-4\underline{\lambda}P_k\beta_{k}\phi_k^2+\frac{48P_k^2\beta_k^2\phi_k^4\bar{\lambda}^2}{\underline{f}^2}+\frac{128P_k^3\beta_k^3\phi_k^6\bar{\lambda}^3}{\underline{f}^3}+\frac{128P_k^4\beta_k^4\phi_k^8\bar{\lambda}^4}{\underline{f}^4}<1$. Hence we have
		\begin{align*}
			E\widetilde{\theta}_k^4
			\leq&\prod_{i=i_1+1}^{k}\left(1-4\underline{\lambda}P_i\beta_{i}\phi_i^2+\frac{48P_i^2\beta_i^2\phi_i^4\bar{\lambda}^2}{\underline{f}^2}\notag\right.\\
			&\left.+\frac{128P_i^3\beta_i^3\phi_i^6\bar{\lambda}^3}{\underline{f}^3}+\frac{128P_i^4\beta_i^4\phi_i^8\bar{\lambda}^4}{\underline{f}^4}\right)E\widetilde{\theta}_{i_1}^4\nonumber\\
			&+\sum_{l=i_1+1}^{k-1}\prod_{i=l+1}^{k}\left(1-4\underline{\lambda}P_i\beta_{i}\phi_i^2+\frac{48P_i^2\beta_i^2\phi_i^4\bar{\lambda}^2}{\underline{f}^2}\notag\right.\\
			&\left.+\frac{128P_i^3\beta_i^3\phi_i^6\bar{\lambda}^3}{\underline{f}^3}+\frac{128P_i^4\beta_i^4\phi_i^8\bar{\lambda}^4}{\underline{f}^4}\right)O\left(\frac{1}{l^2}\right)\nonumber\\
			\leq&E\widetilde{\theta}_{i_1}^2 +\sum_{l=i_1+1}^{k-1}O\left(\frac{1}{l^2}\right)<\infty.
		\end{align*}
		Similarly, when $t=3$, we have $E\widetilde{\theta}_k^6<\infty$. Repeating the above steps, we have $E\widetilde{\theta}_k^{2t}<\infty$ for all $t\geq1$. By Lyapunov 
		inequality, we have $E|\widetilde{\theta}_k|^t<\infty$ for all $t\geq1$, together with (\ref{E_gamma_high}) yields that $E\left[|\gamma_k|^t\right]<\infty$ for all $t\geq1$.

		Now we will prove (\ref{step1}). For $m=1$, according to (\ref{dierxianggaoci}), the boundedness of $E\left[\gamma_k^2\right]$ and Lemma \ref{zenmelaide}, (\ref{Beq1}) can be written as  		
		\begin{align}\label{Ex2}
			E\widetilde{\theta}_k^2
			\leq&\left(1-2\underline{\lambda}P_k\beta_{k}\phi_k^2\right)E\widetilde{\theta}_{k-1}^2+\frac{2P_k^2\bar{\beta}^2\bar{\phi}^2\bar{\lambda}^2}{\underline{f}^2}E\left[\gamma_k^2\right]\nonumber\\
			\leq&\prod_{i=1}^{k}\left(1-P_i\beta_{i}\phi_i^2\right)^{2\underline{\lambda}}E\widetilde{\theta}_0^2\nonumber\\
			&+\sum_{l=1}^{k-1}\prod_{i=l+1}^{k}\left(\left(1-P_i\beta_{i}\phi_i^2\right)^{2\underline{\lambda}}\right)O\left(\frac{1}{l^2}\right)\nonumber\\
			=&
			\left\{
			\begin{aligned}
				&O\left(\frac{1}{k^{2\underline{\lambda}}}\right), &2\underline{\lambda}&<1,\\
				&O\left(\frac{\log k}{k}\right), &2\underline{\lambda}&=1,\\
				&O\left(\frac{1}{k}\right), &2\underline{\lambda}&>1.\\
			\end{aligned}
			\right.
		\end{align}
		For $m=2$, let's come back to (\ref{high_E}). 
		Since $E|\widetilde{\theta}_k|^t<\infty$ for $t\geq1$,  we find that the fourth and fifth terms of (\ref{high_E}) are both as $o\left(\frac{1}{k^2}\right)$.
		For the third term of (\ref{high_E}), we have  $$6E\left[\widetilde{\theta}_{k-1}^2\hat{P}_k^2\hat{\alpha}_k^2\phi_k^2\gamma_k^2(\hat{F}_k-s_k)^2\right]=O\left(\frac{1}{k^2}E\left[\widetilde{\theta}_{k-1}^2\gamma_k^2\right]\right).$$ Next we will give an estimation of $E\left[\widetilde{\theta}_{k-1}^2\gamma_k^2\right]$. By (\ref{gamma_high}), we have $E\left[\widetilde{\theta}_{k-1}^2\gamma_k^2\right]=O\left(E\widetilde{\theta}_{k-1}^2\right)+O\left(E\widetilde{\theta}_{k-1}^4\right)$. 
		For (\ref{Ex4_more}) we know that there exists $i_2$ such that when $i>i_2$, $1-4\underline{\lambda}P_k\beta_{k}\phi_k^2+\frac{48P_k^2\beta_k^2\phi_k^4\bar{\lambda}^2}{\underline{f}^2}+\frac{128P_k^3\beta_k^3\phi_k^6\bar{\lambda}^3}{\underline{f}^3}+\frac{128P_k^4\beta_k^4\phi_k^8\bar{\lambda}^4}{\underline{f}^4}<1-3\underline{\lambda}P_k\beta_k\phi_k^2$. So by Lemma \ref{zenmelaide}, we get
		\begin{align}\label{Ex4_compare}
			E\widetilde{\theta}_k^4
			\leq&\prod_{i=i_2}^{k}\left(1-P_i\beta_{i}\phi_i^2\right)^{3\underline{\lambda}}E\widetilde{\theta}_{i_2-1}^2\nonumber\\
			&+\sum_{l=i_2}^{k-1}\prod_{i=l+1}^{k}\left(\left(1-P_i\beta_{i}\phi_i^2\right)^{3\underline{\lambda}}\right)O\left(\frac{1}{l^2}\right)\nonumber\\
			=&
			\left\{
			\begin{aligned}
				&O\left(\frac{1}{k^{3\underline{\lambda}}}\right), &3\underline{\lambda}&<1,\\
				&O\left(\frac{\log k}{k}\right), &3\underline{\lambda}&=1,\\
				&O\left(\frac{1}{k}\right), &3\underline{\lambda}&>1.\\
			\end{aligned}
			\right.
		\end{align}
		Combining (\ref{Ex2}) and (\ref{Ex4_compare}), we know that 
		\begin{align}\label{Etheta_gamma}
			E\left[\widetilde{\theta}_{k-1}^2\gamma_k^2\right]=
			\left\{
			\begin{aligned}
				&O\left(\frac{1}{k^{2\underline{\lambda}}}\right), &2\underline{\lambda}&<1,\\
				&O\left(\frac{\log k}{k}\right), &2\underline{\lambda}&=1,\\
				&O\left(\frac{1}{k}\right), &2\underline{\lambda}&>1.\\
			\end{aligned}
			\right.
		\end{align}		
		For the second term of (\ref{high_E}), similarly to the proof of (\ref{mean-value}), there exists $\acute{f}_k$ such that $4E\widetilde{\theta}_{k-1}^3\hat{P}_k\hat{\alpha}_k\phi_k\gamma_k(\hat{F}_k-s_k)=-4E\widetilde{\theta}_{k-1}^3\hat{P}_k\hat{\alpha}_k\phi_k\gamma_k\widetilde{z}_k\acute{f}_k$. 
		By (\ref{Etheta_gamma}) and Lemma \ref{zenmelaide}, (\ref{high_E}) can be written as 
		\begin{align}\label{Ex4}
			E\widetilde{\theta}_k^4
			\leq&E\widetilde{\theta}_{k-1}^4-4E\widetilde{\theta}_{k-1}^4P_k\beta_k\lambda_k\phi_k^2+O\left(\frac{1}{k^2}E\left[\widetilde{\theta}_{k-1}^2\gamma_k^2\right]\right)\nonumber\\
			\leq&\prod_{i=1}^{k}(1-P_i\beta_{i}\phi_i^2)^{4\underline{\lambda}}E\widetilde{\theta}_{0}^4\nonumber\\
			&+\sum_{l=1}^{k-1}\prod_{i=l+1}^{k}(1-P_i\beta_{i}\phi_i^2)^{4\underline{\lambda}}O\left(\frac{1}{l^2}E\left[\widetilde{\theta}_{l-1}^2\gamma_l^2\right]\right)\nonumber\\
			=&\left\{
			\begin{aligned}
				&O\left(\frac{1}{k^{4\underline{\lambda}}}\right), &2\underline{\lambda}&<1,\\
				&O\left(\left(\frac{\log k}{k}\right)^2\right), &2\underline{\lambda}&=1,\\
				&O\left(\frac{1}{k^2}\right), &2\underline{\lambda}&>1.\\
			\end{aligned}
			\right.   		
		\end{align}
		Similarly, for $m\geq3$,
		\begin{align}\label{E2m}
			E\widetilde{\theta}_k^{2m}=
			\left\{
			\begin{aligned}
				&O\left(\frac{1}{k^{2m\underline{\lambda}}}\right), &2\underline{\lambda}&<1,\\
				&O\left(\left(\frac{\log k}{k}\right)^m\right), &2\underline{\lambda}&=1,\\
				&O\left(\frac{1}{k^m}\right), &2\underline{\lambda}&>1.\\
			\end{aligned}
			\right.
		\end{align}
		Thus, Step 1 is proved  with $\mu=\underline{\lambda}$.

		\textbf{Step 2:} We will prove $E\widetilde{\theta}_k^2=O\left(\frac{1}{k}\right)$ by discussing the different situations of $\underline{\lambda}$. By (\ref{Ex2}), we have
		\begin{equation*}
			\begin{aligned}
				E\widetilde{\theta}_k^{2}=
				\left\{
				\begin{aligned}
					&O\left(\frac{1}{k^{2\underline{\lambda}}}\right), &2\underline{\lambda}&<1,\\
					&O\left(\frac{\log k}{k}\right), &2\underline{\lambda}&=1,\\
					&O\left(\frac{1}{k}\right), &2\underline{\lambda}&>1.\\
				\end{aligned}
				\right.\\
			\end{aligned}
		\end{equation*}
		 By the proof of Step 1, we know that $\underline{\lambda}>0$. Hence for every $\underline{\lambda}$,  there exists $m\geq1$ such that $2m\underline{\lambda}>1, 2(m-1)\underline{\lambda}\leq1$.
		
		\textbf{Case 1}\quad if $m=1$, i.e. $2\underline{\lambda}>1,$ then $E\widetilde{\theta}_k^2=O\left(\frac{1}{k}\right)$.
		
		\textbf{Case 2}\quad if $m=2$, i.e. $4\underline{\lambda}>1$, $2\underline{\lambda}\leq1$.
		
		If $2\underline{\lambda}=1$, there exists $\delta\in(0,\frac{1}{2})$ such that $E\widetilde{\theta}_k^2=O(\frac{1}{k^{1-\delta}})$. Hence, according to (\ref{Ex4}), we have 
		\begin{equation*}
			\begin{aligned}
				E\widetilde{\theta}_k^4\leq&\prod_{i=1}^{k}(1-P_i\beta_{i}\phi_i^2)^{4\underline{\lambda}}E\widetilde{\theta}_0^4\\
				&+\sum_{l=1}^{k-1}\prod_{i=l+1}^{k}(1-P_i\beta_{i}\phi_i^2)^{4\underline{\lambda}}\frac{1}{l^2}E\widetilde{\theta}_{l-1}^2
				=O\left(\frac{1}{k^{2-\delta}}\right).
			\end{aligned}
		\end{equation*}
		We will prove that $|1-\lambda_k|=O\left(|\widetilde{\theta}_{k-1}|+\big|\frac{1}{k}\sum_{i=1}^{k}\widetilde{\theta}_{i-1}\big|\right)$. First we get
		\begin{align*}
			|1-\lambda_k|
			=&\frac{1}{\hat{f}_k}(\hat{f}_k-\check{f}_k)+\hat{\alpha}_k\check{f}_k\left(\frac{\beta_k-\hat{\beta}_k}{\beta_k\hat{\beta}_k}\right)\\&+\frac{\hat{P}_k\hat{\alpha}_k}{\beta_k}\check{f}_k\left(\frac{1}{\hat{P}_k}-\frac{1}{P_k}\right).
		\end{align*}
		By Mean-Value theorem in \cite{Apostol:105425}, there exists $\breve{\zeta}_k$ between $C-z_k$ and $\zeta_{k-1}$ such that  $\big|\frac{1}{\hat{f}_k}(\hat{f}_k-\check{f}_k)\big|=\big|\frac{f'(\breve{\zeta}_k)}{\hat{f}_k}\left(C-z_k-\zeta_{k-1}\right)\big|\leq\big|\frac{f'(\breve{\zeta}_k)}{\hat{f}_k}\phi_k\widetilde{\theta}_{k-1}\big|=O\left(|\widetilde{\theta}_{k-1}|\right)$.
		Define $\beta(x)=\frac{f^2(x)}{F(x)(1-F(x))}$.
		There exists $\check{\zeta}_k$ between $C-\phi_k\theta$ and $C-z_k$ such that  $\beta_k-\hat{\beta}_k=\beta(C-\phi_k\theta)-\beta(C-z_k)=-\beta'(\check{\zeta}_k)\widetilde{z}_k$. The boundedness of $\phi_k\theta$ and $z_k$ yields that $\beta'(\check{\zeta}_k)$ is bounded. 
		Then we have  $|\beta_k-\hat{\beta}_k|=|-\beta'(\check{\zeta}_k)\widetilde{z}_k|\leq|\beta'(\check{\zeta}_k)|\cdot|\phi_k\widetilde{\theta}_{k-1}|=O\left(|\widetilde{\theta}_{k-1}|\right)$. Similarly, we have 
		$\hat{P}_k(\frac{1}{\hat{P}_k}-\frac{1}{P_k})=O\left(\frac{1}{k}\sum_{i=1}^{k}|\widetilde{\theta}_{i-1}|\right)$. Hence $1-\lambda_k=O\left(|\widetilde{\theta}_{k-1}|+\frac{1}{k}\sum_{i=1}^{k}|\widetilde{\theta}_{i-1}|\right)$. 
	    Since $E\widetilde{\theta}_k^2=O(\frac{1}{k^{1-\delta}})$, $\frac{1-\delta}{2}<1$, we have $\frac{1}{k}\sum_{i=1}^{k}\sqrt{E\widetilde{\theta}_{i-1}^2}=O\left(\frac{1}{k}\sum_{i=1}^{k}\left(\frac{1}{i}\right)^{\frac{1-\delta}{2}}\right)=O\left(\frac{1}{k^{(1-\delta)/2}}\right)$. Then, (\ref{Ex2}) can be rewritten as		
		\begin{align}\label{jiasu}
			E\widetilde{\theta}_k^2
			=&(1-2P_k\beta_{k}\phi_k^2)E\widetilde{\theta}_{k-1}^2+O\left(\frac{1}{k}\right)\nonumber\\
			&\cdot E\left[\left(|\widetilde{\theta}_{k-1}|+\frac{\sum_{i=1}^{k}|\widetilde{\theta}_{i-1}|}{k}\right)\widetilde{\theta}_{k-1}^2\right]+O\left(\frac{1}{k^2}\right)\nonumber\\
			\leq&(1-2P_k\beta_{k}\phi_k^2)E\widetilde{\theta}_{k-1}^2+O\left(\frac{1}{k}\right)E\left[|\widetilde{\theta}_{k-1}|\widetilde{\theta}_{k-1}^2\right]\nonumber\\
			&+O\left(\frac{1}{k}\right)\sum_{i=1}^{k}E\left[\frac{|\widetilde{\theta}_{i-1}|}{k}\widetilde{\theta}_{k-1}^2\right]+O\left(\frac{1}{k^2}\right)\nonumber\\
			\leq&(1-2P_k\beta_{k}\phi_k^2)E\widetilde{\theta}_{k-1}^2+O\left(\frac{1}{k}\right)\sqrt{E\widetilde{\theta}_{k-1}^2E\widetilde{\theta}_{k-1}^4}\nonumber\\
			&+O\left(\frac{1}{k}\right)\sqrt{E\widetilde{\theta}_{k-1}^4}\cdot\frac{1}{k}\sum_{i=1}^{k}\sqrt{E\widetilde{\theta}_{i-1}^2}+O\left(\frac{1}{k^2}\right)\nonumber\\
			\leq& (1-2P_k\beta_{k}\phi_k^2)E\widetilde{\theta}_{k-1}^2\nonumber\displaybreak\\
			&+O\left(\frac{1}{k^{1+(1-\delta)/2+(2-\delta)/2}}\right)+O\left(\frac{1}{k^2}\right)\nonumber\\
			=&\prod_{i=1}^{k}(1-P_i\beta_{i}\phi_i^2)^2E\widetilde{\theta}_0^2+\sum_{l=1}^{k-1}\prod_{i=l+1}^{k}(1-P_i\beta_{i}\phi_i^2)^2O\left(\frac{1}{l^2}\right)\nonumber\\
			=&O\left(\frac{1}{k}\right).
		\end{align}
		If $2\underline{\lambda}<1$, $E\widetilde{\theta}_k^2=O\left(\frac{1}{k^{2\underline{\lambda}}}\right)$. Then $\frac{1}{k}\sum_{i=1}^{k}\sqrt{E\widetilde{\theta}_{i-1}^2}
		=O\left(\frac{1}{k^{\underline{\lambda}}}\right)$. By (\ref{E2m}), we get $E\widetilde{\theta}_k^4=O\left(\frac{1}{k^{4\underline{\lambda}}}\right)$. Substituting it into (\ref{jiasu}), we have 
		\begin{align*}
			E\widetilde{\theta}_k^2\leq&(1-2P_k\beta_{k}\phi_k^2)E\widetilde{\theta}_{k-1}^2+O\left(\frac{1}{k}\right)\sqrt{E\widetilde{\theta}_{k-1}^2E\widetilde{\theta}_{k-1}^4}\\
			&+O\left(\frac{1}{k}\right)\sqrt{E\widetilde{\theta}_{k-1}^4}\cdot\frac{1}{k}\sum_{i=1}^{k}\sqrt{E\widetilde{\theta}_{i-1}^2}+O\left(\frac{1}{k^2}\right)\\
			=&O\left(\frac{1}{k^2}\right)+O\left(\sum_{i=1}^{k}\left(\frac{i}{k}\right)^2\frac{1}{i^{1+2\underline{\lambda}+\underline{\lambda}}}\right)+O\left(\sum_{i=1}^{k}\left(\frac{i}{k}\right)^2\frac{1}{i^2}\right)\\
			=&\left\{
			\begin{aligned}
				&O\left(\frac{1}{k^{3\underline{\lambda}}}\right), &3\underline{\lambda}&<1,\\
				&O\left(\frac{1}{k}\right), &3\underline{\lambda}&\geq1.\\
			\end{aligned}
			\right.\\
		\end{align*}		
		If $3\underline{\lambda}\geq1$, $E\widetilde{\theta}_k^2=O\left(\frac{1}{k}\right)$ holds. If $3\underline{\lambda}<1$, then $E\widetilde{\theta}_k^2=O\left(\frac{1}{k^{3\underline{\lambda}}}\right)$. So $\frac{1}{k}\sum_{i=1}^{k}\sqrt{E\widetilde{\theta}_{i-1}^2}
		=O\left(\frac{1}{k^{3\underline{\lambda}/2}}\right)$. Therefore we have
		\begin{align*}
				E\widetilde{\theta}_k^2\leq&(1-2P_k\beta_{k}\phi_k^2)E\widetilde{\theta}_{k-1}^2+O\left(\frac{1}{k^{1+2\underline{\lambda}}}\right)\\
				&\cdot\left(\sqrt{E\widetilde{\theta}_{k-1}^2}+\frac{1}{k}\sum_{i=1}^{k}\sqrt{E\widetilde{\theta}_{i-1}^2}\right)+O\left(\frac{1}{k^2}\right)\\
				\leq&\prod_{i=1}^{k}\left(1-P_i\beta_{i}\phi_i^2\right)^2E\widetilde{\theta}_0^2+\sum_{l=1}^{k-1}\prod_{i=l+1}^{k}\left(1-\underline{\lambda}P_i\beta_{i}\phi_i^2\right)^2\\
				&\cdot\left(O\left(\frac{1}{l^{1+4\underline{\lambda}-2^{-1}\underline{\lambda}}}\right)+O\left(\frac{1}{l^2}\right)\right)\\
				=&\left\{
				\begin{aligned}
					&O\left(\frac{1}{k^{\frac{7}{2}\underline{\lambda}}}\right), &\frac{7}{2}\underline{\lambda}&<1,\\
					&O\left(\frac{1}{k}\right), &\frac{7}{2}\underline{\lambda}&\geq1.\\
				\end{aligned}
				\right.\\
		\end{align*}
		Repeating the above steps for $t$ times, we have 
		\begin{equation*}
			\begin{aligned}
				E\widetilde{\theta}_k^2=\left\{
				\begin{aligned}
					&O\left(\frac{1}{k^{4\underline{\lambda}-2^{-t}\underline{\lambda}}}\right), &4\underline{\lambda}-2^{-t}\underline{\lambda}&<1,\\
					&O\left(\frac{1}{k}\right), &4\underline{\lambda}-2^{-t}\underline{\lambda}&\geq1.
				\end{aligned}
				\right.
			\end{aligned}
		\end{equation*}
		Since $4\underline{\lambda}>1$, there exists $t_1$, such that $4\underline{\lambda}-2^{-t_1}\underline{\lambda}\geq1$, then  $E\widetilde{\theta}_k^2=O\left(\frac{1}{k}\right)$.

		\textbf{Case 3}\quad If $m=3$, then $6\underline{\lambda}>1$, and $4\underline{\lambda}\leq1$. (\ref{Ex4}) can be rewritten as 
		\begin{align}\label{Ex4jiasu}
			E\widetilde{\theta}_k^4
			\leq&E\widetilde{\theta}_{k-1}^4-4P_k\beta_{k}\phi_k^2E\widetilde{\theta}_{k-1}^4\nonumber\\
			&+4P_k\beta_{k}\phi_k^2E\left[(1-\lambda_k)\widetilde{\theta}_{k-1}^4\right]+O\left(\frac{1}{k^2}E\widetilde{\theta}_{k-1}^2\right)\nonumber\\
			\leq&(1-4P_k\beta_{k}\phi_k^2)E\widetilde{\theta}_{k-1}^4+O\left(\frac{1}{k}\right)E\left[|\widetilde{\theta}_{k-1}|\widetilde{\theta}_{k-1}^4\right]\nonumber\\
			&+O\left(\frac{1}{k}\right)\sum_{i=1}^{k}E\left[\frac{|\widetilde{\theta}_{i-1}|}{k}\widetilde{\theta}_{k-1}^4\right]+O\left(\frac{1}{k^2}E\widetilde{\theta}_{k-1}^2\right)\nonumber\\
			\leq&(1-4P_k\beta_{k}\phi_k^2)E\widetilde{\theta}_{k-1}^4+O\left(\frac{1}{k}\right)\sqrt{E\widetilde{\theta}_{k-1}^4E\widetilde{\theta}_{k-1}^6}\nonumber\\
			&+O\left(\frac{1}{k}\right)\sqrt{E\widetilde{\theta}_{k-1}^6}\sqrt[4]{E\widetilde{\theta}_{k-1}^4}\cdot\frac{1}{k}\sum_{i=1}^{k}\sqrt[4]{E\widetilde{\theta}_{i-1}^4}\nonumber\\
			&+O\left(\frac{1}{k^2}E\widetilde{\theta}_{k-1}^2\right).
		\end{align}
		By (\ref{E2m}), we know that $E\widetilde{\theta}_k^4=O\left(\frac{1}{k^{4\underline{\lambda}}}\right)$, $E\widetilde{\theta}_k^6=O\left(\frac{1}{k^{6\underline{\lambda}}}\right)$. Hence $\frac{1}{k}\sum_{i=1}^{k}\sqrt[4]{E\widetilde{\theta}_{i-1}^4}
		=O\left(\frac{1}{k^{\underline{\lambda}}}\right)$. And we know that $E\widetilde{\theta}_k^2=O\left(\frac{1}{k^{4\underline{\lambda}-2^{-t_1}\underline{\lambda}}}\right)$ since $4\underline{\lambda}-2^{-t_1}\underline{\lambda}<1$. 
		Therefore we know that
		\begin{align}
			E\widetilde{\theta}_k^4
			\leq&\prod_{i=1}^{k}(1-P_k\beta_{k}\phi_k^2)^4E\widetilde{\theta}_0^2+\sum_{l=1}^{k-1}\prod_{i=l+1}^{k}(1-P_k\beta_{k}\phi_k^2)^4\nonumber\\
			&\cdot\left(O\left(\frac{1}{l^{1+5\underline{\lambda}}}\right)+O\left(\frac{1}{l^{2+4\underline{\lambda}-2^{-t_1}\underline{\lambda}}}\right)\right)=O\left(\frac{1}{k^{5\underline{\lambda}}}\right).\nonumber
		\end{align}
		By substituting $E\widetilde{\theta}_k^4=O\left(\frac{1}{k^{5\underline{\lambda}}}\right)$ into (\ref{Ex4jiasu}) and repeating this step for $t_2$ times, we have $E\widetilde{\theta}_k^4=O\left(\frac{1}{k^{6\underline{\lambda}-2^{-t_2}\underline{\lambda}}}\right)$. 
		Then $\frac{1}{k}\sum_{i=1}^{k}\sqrt{E\widetilde{\theta}_{i-1}^2}=O\left(\frac{1}{k^{2\underline{\lambda}-2^{-(t_1+1)}\underline{\lambda}}}\right)$. Hence, 
		\begin{align*}
				E\widetilde{\theta}_k^2\leq&(1-2P_k\beta_{k}\phi_k^2)E\widetilde{\theta}_{k-1}^2+O\left(\frac{1}{k}\right)\sqrt{E\widetilde{\theta}_{k-1}^2E\widetilde{\theta}_{k-1}^4}\\
				&+O\left(\frac{1}{k}\right)\sqrt{E\widetilde{\theta}_{k-1}^4}\cdot\frac{1}{k}\sum_{i=1}^{k}\sqrt{E\widetilde{\theta}_{i-1}^2}+O\left(\frac{1}{k^2}\right)\\
				=&O\left(\frac{1}{k^2}\right)+O\left(\sum_{i=1}^{k}\left(\frac{i}{k}\right)^2\frac{1}{i^{1+3\underline{\lambda}-2^{-(t_2+1)}\underline{\lambda}+2\underline{\lambda}-2^{-(t_1+1)}\underline{\lambda}}}\right)\\
				&+O\left(\sum_{i=1}^{k}\left(\frac{i}{k}\right)^2\frac{1}{i^2}\right)\displaybreak\\
				=&\left\{
				\begin{aligned}
					&O\left(\frac{1}{k^{v_1}}\right), &v_1&<1,\\
					&O\left(\frac{1}{k}\right), &v_1&\geq1,
				\end{aligned}
				\right.
		\end{align*}
		where $v_1=5\underline{\lambda}-2^{-(t_2+1)}\underline{\lambda}-2^{-(t_1+1)}\underline{\lambda}$. If $v_1\geq1$, then $E\widetilde{\theta}_k^2=O\left(\frac{1}{k}\right)$ holds. Otherwise, substituting $E\widetilde{\theta}_k^2=O\left(\frac{1}{k^{v_1}}\right)$ into the equation above and repeating the above steps for $t_3$ times, we have 
		\begin{equation*}
			\begin{aligned}
				E\widetilde{\theta}_k^2=\left\{
				\begin{aligned}
					&O\left(\frac{1}{k^{v_2}}\right), &v_2&<1,\\
					&O\left(\frac{1}{k}\right), &v_2&\geq1,\\
				\end{aligned}
				\right.\\
			\end{aligned}
		\end{equation*}
		where $v_2=6\underline{\lambda}-2^{-t_2}\underline{\lambda}-2^{-t_3}(2\underline{\lambda}+2^{-t_1}\underline{\lambda}-2^{-t_2}\underline{\lambda})$. Since $6\underline{\lambda}>1$, there exists $t_1,t_2,t_3$, such that $v_2\geq1$. $E\widetilde{\theta}_k^2=O\left(\frac{1}{k}\right)$ holds.
		
		When $m>3$, similarly we can prove that $E\widetilde{\theta}_k^2=O\left(\frac{1}{k}\right)$.
		Finally, we prove that $E\widetilde{\theta}_k^2=O\left(\frac{1}{k}\right)$.
		\end{proof}
	
	The following theorem shows the asymptotic efficiency of the IMPF algorithm. In other words, the covariance of estimation error converges to the CR lower bound $\Delta_k$, where $\Delta_k=\left(\sum_{i=1}^{k}\frac{f_i^2}{F_i(1-F_i)}\phi_i\phi_i^T\right)^{-1}$ is given in Lemma \ref{lemmaefficiency} in Appendix.
	\begin{theorem}\label{theorem efficiency}
		For the first-order FIR system (\ref{M}) with binary-valued observations (\ref{moxingtwo}), if Assumptions \ref{A1}-\ref{A3} hold, the IMPF algorithm is asymptotically efficient, i.e.,
		\begin{align*}
			\lim\limits_{k\to\infty}\Delta_k^{-1}(E\widetilde{\theta}_k^2-\Delta_k)=0,
		\end{align*}
	where $\Delta_k=\left(\sum_{i=1}^{k}\frac{f_i^2}{F_i(1-F_i)}\phi_i\phi_i^T\right)^{-1}.$ 
	\end{theorem}
	\begin{proof}
		Similarly to the proof of Theorem \ref{theorem1}, according to Mean-Value theorem in \cite{Apostol:105425}, we have	
			\begin{align*}
				|\widetilde{\theta}_k|=&|\widetilde{\theta}_{k-1}+\hat{P}_k\hat{\alpha}_k\phi_k\gamma_k(\hat{F}_k-F_k+F_k-s_k)|\\
				=&|\widetilde{\theta}_{k-1}-\hat{P}_k\hat{\alpha}_k\phi_k\widetilde{z}_k\gamma_k\check{f}_k+\hat{P}_k\hat{\alpha}_k\phi_k\gamma_k(F_k-s_k)|\\
				=&|\widetilde{\theta}_{k-1}(1-\hat{P}_k\hat{\alpha}_k\phi_k^2\check{f}_k)+\hat{P}_k\hat{\alpha}_k\phi_k\check{f}_k(\phi_k\widetilde{\theta}_{k-1}-\widetilde{z}_k\gamma_k)\\
				&+\hat{P}_k\hat{\alpha}_k\phi_k\gamma_k(F_k-s_k)|\\
				=&|\widetilde{\theta}_{k-1}(1-\hat{P}_k\hat{\alpha}_k\phi_k^2\hat{f}_k)+\widetilde{\theta}_{k-1}\hat{P}_k\hat{\alpha}_k\phi_k^2(\hat{f}_k-\check{f}_k)\\
				&+\hat{P}_k\hat{\alpha}_k\phi_k\check{f}_k(\phi_k\widetilde{\theta}_{k-1}-\widetilde{z}_k\gamma_k)+\hat{P}_k\hat{\alpha}_k\phi_k\gamma_k(F_k-s_k)|\\
				=&\bigg|\widetilde{\theta}_{k-1}(1-\frac{\hat{\beta}_k\hat{P}_k\phi_k^2}{1+\hat{\beta}_k\hat{P}_k\phi_k^2})-\widetilde{\theta}_{k-1}\frac{(\hat{\beta}_k\hat{P}_k\phi_k^2)^2}{1+\hat{\beta}_k\hat{P}_k\phi_k^2}\\
				&+\widetilde{\theta}_{k-1}\hat{P}_k\hat{\alpha}_k\phi_k^2(\hat{f}_k-\check{f}_k)+\hat{P}_k\hat{\alpha}_k\phi_k\check{f}_k(\phi_k\widetilde{\theta}_{k-1}-\widetilde{z}_k\gamma_k)\\
				&+\hat{P}_k\hat{\alpha}_k\phi_k\gamma_k(F_k-s_k)\bigg|,
			\end{align*}
		where $\check{f}_k=f(\zeta_{k-1})$ and $\widetilde{z}_k$ are the same with  (\ref{shelambda}). 
		We know that $E\big|\widetilde{\theta}_{k-1}\frac{(\hat{\beta}_k\hat{P}_{k-1}\phi_k^2)^2}{1+\hat{\beta}_k\hat{P}_{k-1}\phi_k^2}\big|=o\left(\frac{1}{k^2}\right)$
		and $\hat{P}_k\hat{P}_{k-1}^{-1}=1-\frac{\hat{\beta}_k\hat{P}_{k-1}\phi_k^2}{1+\hat{\beta}_k\hat{P}_{k-1}\phi_k^2}$. Combining with $E(F_k-s_k)^2=F_k(1-F_k)$, it can be seen that
		\begin{align*}
			E\widetilde{\theta}_k^2			=&E\left[\widetilde{\theta}_{k-1}\left(1-\frac{\hat{\beta}_k\hat{P}_k\phi_k^2}{1+\hat{\beta}_k\hat{P}_k\phi_k^2}\right)-\widetilde{\theta}_{k-1}\frac{(\hat{\beta}_k\hat{P}_k\phi_k^2)^2}{1+\hat{\beta}_k\hat{P}_k\phi_k^2}\right.\\
			&+\widetilde{\theta}_{k-1}\hat{P}_k\hat{\alpha}_k\phi_k^2(\hat{f}_k-\check{f}_k)+\hat{P}_k\hat{\alpha}_k\phi_k\check{f}_k(\phi_k\widetilde{\theta}_{k-1}-\widetilde{z}_k\gamma_k)\\
			&\left.+\hat{P}_k\hat{\alpha}_k\phi_k\gamma_k(F_k-s_k)\right]^2\\
			\leq&E(\widetilde{\theta}_{k-1}\hat{P}_k\hat{P}_{k-1}^{-1})^2+E\hat{P}_k^2\hat{\alpha}_k^2\phi_k^2\check{f}_k^2(\phi_k\widetilde{\theta}_{k-1}-\widetilde{z}_k\gamma_k)^2\\
			&+F_k(1-F_k)\phi_k^2E\hat{P}_k^2\hat{\alpha}_k^2\gamma_k^2\\
			&+2E|\widetilde{\theta}_{k-1}^2\hat{P}_{k}\hat{\alpha}_k\phi_k^2(\hat{f}_k-\check{f}_k)|\\
			&+2E|\widetilde{\theta}_{k-1}\hat{P}_{k}\hat{\alpha}_k\phi_k\check{f}_k(\phi_k\widetilde{\theta}_{k-1}-\widetilde{z}_k\gamma_k)|\\
			&+2E\widetilde{\theta}_{k-1}\hat{P}_k^2\hat{\alpha}_k^2\phi_k^3\check{f}_k(\hat{f}_k-\check{f}_k)(\phi_k\widetilde{\theta}_{k-1}-\widetilde{z}_k\gamma_k)\\
			&+o\left(\frac{1}{k^2}\right).
		\end{align*}
		From Assumption \ref{A2}, Mean-Value theorem in \cite{Apostol:105425} and the boundedness of $z_k$, we have  $|E(\hat{f}_k-\check{f}_k)|=O(E|\widetilde{\theta}_{k-1}|)\leq O(\sqrt{E\widetilde{\theta}_{k-1}^2})=O(\frac{1}{\sqrt{k}})$.
		Since $P_k=O\left(\frac{1}{k}\right)$ and  $E\widetilde{\theta}_k^2=O\left(\frac{1}{k}\right)$, we have
		\begin{align}\label{zuiyou1}
			E\widetilde{\theta}_k^2\leq&E(\widetilde{\theta}_{k-1}\hat{P}_k\hat{P}_{k-1}^{-1})^2+E\hat{P}_k^2\hat{\alpha}_k^2\phi_k^2\check{f}_k^2(\phi_k\widetilde{\theta}_{k-1}-\widetilde{z}_k\gamma_k)^2\nonumber\\
			&+F_k(1-F_k)\phi_k^2E\hat{P}_k^2\hat{\alpha}_k^2\gamma_k^2\nonumber\\
			&+2E\widetilde{\theta}_{k-1}\hat{P}_{k}\hat{\alpha}_k\phi_k\check{f}_k(\phi_k\widetilde{\theta}_{k-1}-\widetilde{z}_k\gamma_k)\nonumber\\
			&+o\left(\frac{1}{k^2}\right).
		\end{align}		
		Similarly to the proof of Theorem \ref{1k}, 
		we know $|\beta_k-\hat{\beta}_k|=O(|\widetilde{\theta}_{k-1}|)$ and  $|\alpha_k-\hat{\alpha}_k|=O(|\widetilde{\theta}_{k-1}|)$.
		Then for the first term on RHS of (\ref{zuiyou1}), we have		
		\begin{align}\label{high_1}
			&E(\widetilde{\theta}_{k-1}\hat{P}_k\hat{P}_{k-1}^{-1})^2\nonumber\\
			=&E\widetilde{\theta}_{k-1}^2P_k^2P_{k-1}^{-2}+E\widetilde{\theta}_{k-1}^2(\hat{P}_k\hat{P}_{k-1}^{-1}+P_kP_{k-1}^{-1})\nonumber\\
			&\cdot\left(\frac{P_{k-1}\beta_k\phi_k^2}{1+P_{k-1}\beta_k\phi_k^2}-\frac{\hat{P}_{k-1}\hat{\beta}_k\phi_k^2}{1+\hat{P}_{k-1}\hat{\beta}_k\phi_k^2}\right)\nonumber\\
			=&E\widetilde{\theta}_{k-1}^2P_k^2P_{k-1}^{-2}+E\widetilde{\theta}_{k-1}^2(\beta_k-\hat{\beta}_k)O\left(\frac{1}{k}\right)\nonumber\\
			=&E\widetilde{\theta}_{k-1}^2P_k^2P_{k-1}^{-2}+o\left(\frac{1}{k^2}\right).
		\end{align}		
		According to Markov's Inequality, we find that 		
		\begin{align}\label{301}
			EI_{\{|\phi_k\hat{\theta}_{k-1}|>M\}}^2
			=&P\{|\phi_k\hat{\theta}_{k-1}|>M\}\nonumber\\
			\leq& P\{|\phi_k\widetilde{\theta}_{k-1}|>M-\bar{\phi}\bar{\theta}\}\nonumber\displaybreak\\
			\leq& \frac{1}{2}E(\phi_k\widetilde{\theta}_{k-1})^2\nonumber\\
			=&O\left(\frac{1}{k}\right).
		\end{align}		
		For the second term on RHS of (\ref{zuiyou1}), because of the boundedness of $\hat{\alpha}_k$, $\phi_k$, $f$, $E|\widetilde{\theta}_k|^t$ and $E|\hat{\theta}_k|^t$, we have 
		\begin{align}\label{high_2}
			&E\hat{P}_k^2\hat{\alpha}_k^2\phi_k^2\check{f}_k^2(\phi_k\widetilde{\theta}_{k-1}-\widetilde{z}_k\gamma_k)^2\nonumber\\
			\leq&O\left(\frac{1}{k^2}\right)E\left[(\widetilde{\theta}_{k-1}^2+|\hat{\theta}_{k-1}|^2)I_{\{|\phi_k\hat{\theta}_{k-1}|>M\}}\right]\nonumber\\
			\leq&O\left(\frac{1}{k^2}\right)\sqrt{E\left(\widetilde{\theta}_{k-1}^2+|\hat{\theta}_{k-1}|^2\right)^2}\sqrt{EI_{\{|\phi_k\hat{\theta}_{k-1}|>M\}}^2}\nonumber\\
			=&o\left(\frac{1}{k^2}\right).
		\end{align}
		For the third term on RHS of (\ref{zuiyou1}), we have 		
		\begin{align}\label{high_3_1}
			&F_k(1-F_k)\phi_k^2E\hat{P}_k^2\hat{\alpha}_k^2\gamma_k^2\nonumber\\
			=&F_k(1-F_k)\phi_k^2EP_k^2\alpha_k^2\gamma_k^2\nonumber\\
			&+F_k(1-F_k)\phi_k^2E(\hat{P}_k^2\hat{\alpha}_k^2-P_k^2\alpha_k^2)\gamma_k^2.
		\end{align}
		Based on Lyapunov inequality, (\ref{gamma_high}) and (\ref{E2m}), we know that
		\begin{align}\label{high_3_2}
			&F_k(1-F_k)\phi_k^2E(\hat{P}_k^2\hat{\alpha}_k^2-P_k^2\alpha_k^2)\gamma_k^2\nonumber\\
			=&F_k(1-F_k)O\left(\frac{1}{k^2}\right)E[(\hat{\alpha}_k-\alpha_k)(\hat{\alpha}_k+\alpha_k)\gamma_k^2]\nonumber\\
			=&o\left(\frac{1}{k^2}\right).
		\end{align}
		For the fourth term on RHS of (\ref{zuiyou1}), we get
		\begin{align}\label{zuiyou}
			&2E\widetilde{\theta}_{k-1}\hat{P}_{k}\hat{\alpha}_k\phi_k\check{f}_k(\phi_k\widetilde{\theta}_{k-1}-\widetilde{z}_k\gamma_k)\nonumber\\
			=&O\left(\frac{1}{k}\right)E\left[(\phi_k\widetilde{\theta}_{k-1}^2-\widetilde{\theta}_{k-1}\widetilde{z}_k\gamma_k)I_{\{|\phi_k\hat{\theta}_{k-1}|>M\}}\right]\nonumber\\
			\leq& O\left(\frac{1}{k}\right)\sqrt{E(\phi_k\widetilde{\theta}_{k-1}^2-\widetilde{\theta}_{k-1}\widetilde{z}_k\gamma_k)^2}\sqrt{EI_{\{|\phi_k\hat{\theta}_{k-1}|>M\}}^2}.
		\end{align}
		By $|\widetilde{z}_k|\leq2M$ and (\ref{gamma_high}), we have
		\begin{align}\label{302}
			&E(\phi_k\widetilde{\theta}_{k-1}^2-\widetilde{\theta}_{k-1}\widetilde{z}_k\gamma_k)^2\nonumber\\
			=&E\left[(\phi_k\widetilde{\theta}_{k-1}^2-\widetilde{\theta}_{k-1}\widetilde{z}_k\gamma_k)I_{\{|\phi_k\hat{\theta}_{k-1}|>M\}}\right]^2\nonumber\\
			=&o\left(\frac{1}{k}\right).
		\end{align}
		Substituting (\ref{301}) and (\ref{302}) into (\ref{zuiyou}), we have
		\begin{align}\label{high_4}
			2E\widetilde{\theta}_{k-1}\hat{P}_{k}\hat{\alpha}_k\phi_k\check{f}_k(\phi_k\widetilde{\theta}_{k-1}-\widetilde{z}_k\gamma_k)
			=o\left(\frac{1}{k^2}\right).
		\end{align} 
		Let's give a more precise estimate of $E\gamma_k^2$. From (\ref{E2m}), there exists $\bar{M}$, such that $E|\phi_k\hat{\theta}_{k-1}|^4
		\leq\bar{M}$. Hence 
		\begin{align}\label{gamma_precise}
			E\gamma_k^2\leq&1+E\left[|\phi_k\hat{\theta}_{k-1}|^2I_{\{|\phi_k\hat{\theta}_{k-1}|>M\}}\right]\nonumber\\
			\leq&1+\sqrt{\bar{M}}\sqrt{EI^2_{\{|\phi_k\hat{\theta}_{k-1}|>M\}}}\nonumber\\
			=&1+O\left(\frac{1}{k^{1/2}}\right).
		\end{align}
		According to Lemma \ref{lemmaefficiency} and the definition of $P_k$, we know that $\Delta_k=P_k=\left(\sum_{i=1}^{k}\phi_i^2\beta_i\right)^{-1}$. Then substituting (\ref{high_1}), (\ref{high_2})-(\ref{high_3_2}),  (\ref{high_4}) and (\ref{gamma_precise}) into (\ref{zuiyou1}), 
		it can be seen that
		\begin{align*}
			E\widetilde{\theta}_k^2\leq& E\widetilde{\theta}_{k-1}^2P_k^2P_{k-1}^{-2}+F_k(1-F_k)\phi_k^2P_k^2\alpha_k^2E\gamma_k^2+o\left(\frac{1}{k^2}\right)\\
			=&E\widetilde{\theta}_{k-1}^2P_k^2P_{k-1}^{-2}+\phi_k^2P_k^2\frac{f_k^2}{F_k(1-F_k)}+o\left(\frac{1}{k^2}\right)\\
			\leq&P_k^2P_0^{-2}E\widetilde{\theta}_0^2+P_k^2\sum_{i=1}^{k}\phi_i^2\frac{f_i^2}{F_i(1-F_i)}\\
			&+P_k^2\sum_{i=1}^{k}P_{i-1}^{-2}o(\frac{1}{(i-1)^2})+o\left(\frac{1}{k^2}\right)\\
			=&O\left(\frac{1}{k^2}\right)+P_k^2\sum_{i=1}^{k}\phi_i^2\beta_i+O\left(\frac{1}{k^2}\right)\sum_{i=1}^{k}o(1)+o\left(\frac{1}{k^2}\right)\\
			\sim& P_k^2\sum_{i=1}^{k}\phi_i^2\beta_{i}+o\left(\frac{1}{k}\right)\\
			\sim& \Delta_k+o\left(\frac{1}{k}\right).
		\end{align*}	
		Then we prove that   $\lim\limits_{k\to\infty}\Delta_k^{-1}(E\widetilde{\theta}_k^2-\Delta_k)=0$.
	\end{proof}

\section{Simulation}
In this section, we will illustrate the theoretical results with three examples. Example 1 verifies the convergence properties of the RPFI algorithm 
Example 2 verifies the asymptotically efficiency of IMPF algorithm and makes a comparison between the IMPF algorithm, the RPFI algorithm and the algorithm in \cite{Ke2022RecursiveIO}. Example 3 shows the asymptotically efficiency for high-order FIR systems and compare the computational complexity of the IMPF algorithm and the algorithm in \cite{10247592}, \cite{wang2023asymptotically} and \cite{Ke2022RecursiveIO}.

\textbf{Example 1:}
Consider a third-order FIR system 
\begin{align*}
	y_{k+1}=\phi_k^T\theta+d_{k+1},k\geq0,
\end{align*}
whose quantized output is $s_{k+1}=I_{\{y_{k+1}\leq C\}}$, 
where $C=0.8$ is a known constant threshold; the system noise $\{d_k,k\geq1\}$ satisfied Assumption \ref{A2}; the constant parameter $\theta=[0.1,0.5,0.9]^T$ is unknown, but the prior information $[-1,1]\times[-1,1]\times[-1,1]$ can be provided; the inputs are periodically generated by $\phi_{3k}=[2,0,1]$, $\phi_{3k+1}=[1,2,0]$, $\phi_{3k+2}=[0,1,2]$; the initial value $\hat{\theta}_0=[0.3,0.3,0.3]^T$.
\begin{figure}[H]
	\centering \includegraphics[width=0.75\columnwidth]{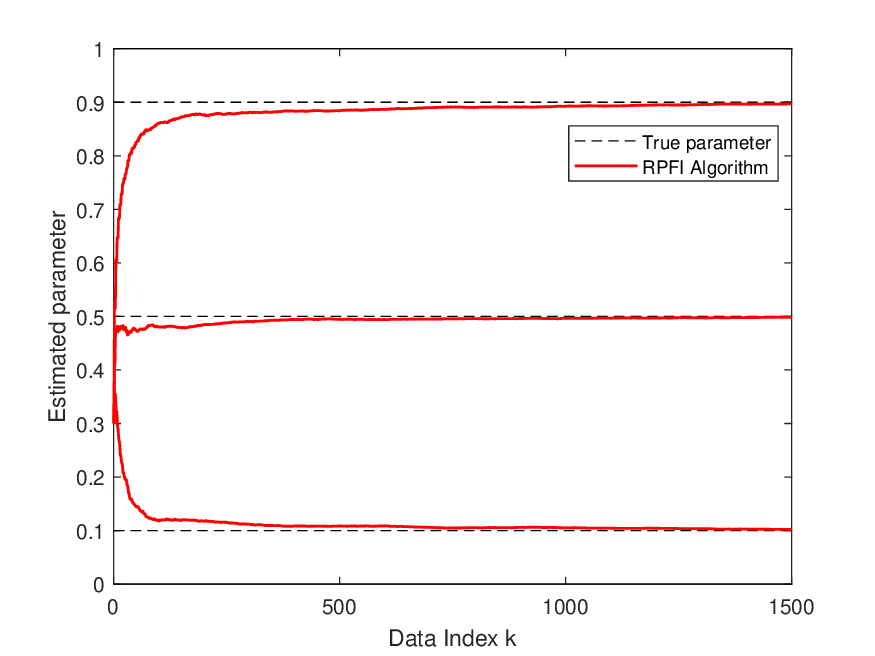}
	\caption{The convergence of the RPFI algorithm}
	\label{figalmostsurely}
\end{figure}
\begin{figure}[H]
	\centering \includegraphics[width=0.75\columnwidth]{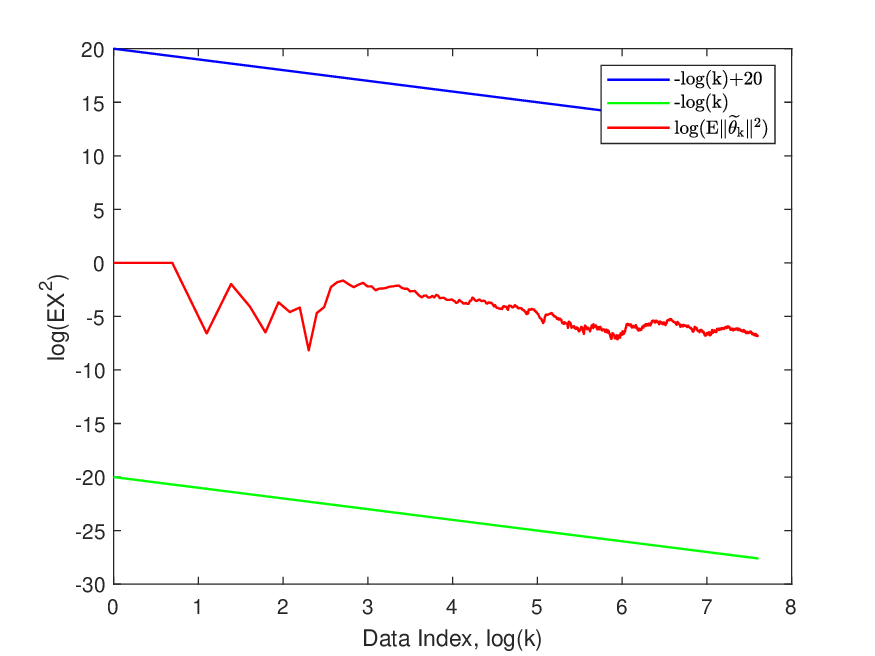}
	\caption{Convergence rate of the RPFI algorithm}
	\label{figspeed}
\end{figure}
Figure \ref{figalmostsurely} shows that the RPFI algorithm converges to the true parameter almost surely. Figure \ref{figspeed} shows that the RPFI algorithm has a mean square convergence rate of O($\frac{1}{k}$).


%
%
%

\textbf{Example 2:}  Consider a first-order FIR system 
\begin{align*}
	y_{k+1}=\phi_k\theta+d_{k+1},k\geq0,
\end{align*}
whose quantized output is $s_{k+1}=I_{\{y_{k+1}\leq C\}}$, 
where $C=4$ is a known constant threshold; the system noise $\{d_k,k\geq1\}$ satisfies Assumption \ref{A2}; the constant parameter $\theta=3$ is unknown, but the prior information $\theta\in[-5,5]$ can be provided; the inputs follow $|\phi_k|\leq3$ and $\phi_k$ is randomly generated in the interval $(1,3)$; and the initial value $\hat{\theta}_0=1$ and $P_0=1$.
\begin{figure}[H]
	\centering \includegraphics[width=0.75\columnwidth]{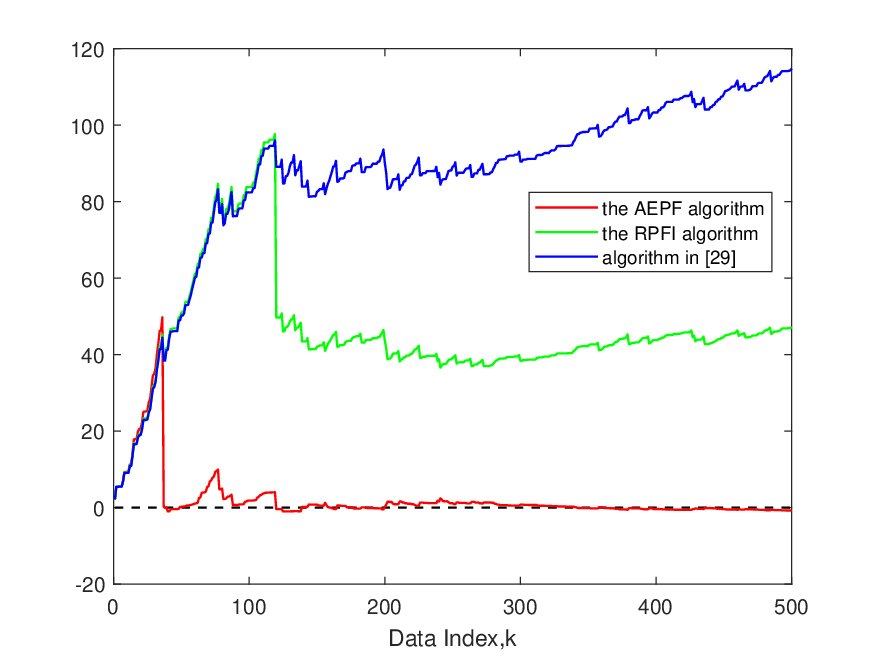}
	\caption{$\Delta_k^{-1}(E\widetilde{\theta}_k^2-\Delta_k)$ of the IMPF algorithm, RPFI algorithm and the algorithm in \cite{Ke2022RecursiveIO}}
	\label{figefficiency}
\end{figure}
Figure \ref{figefficiency} shows that $\Delta_k^{-1}(E\widetilde{\theta}_k^2-\Delta_k)$ of the IMPF algorithm converges to zero 
while the algorithm in \cite{Ke2022RecursiveIO} can not. This verifies that the IMPF algorithm is asymptotically efficient while the algorithm in \cite{Ke2022RecursiveIO} is not. What's more, it can be seen that although the RPFI algorithm is not asymptotically efficient, it performs better than the algorithm in \cite{Ke2022RecursiveIO}. When the data index reaches around 120, the curve of the RPFI algorithm shows a rapid decline. That's because the accelerated coefficient can quickly pull back the estimate when it deviates from prior information. Fully utilizing prior information is another advantage of RPFI algorithm and IMPF algorithm compared to the algorithm in \cite{Ke2022RecursiveIO}.

\textbf{Example 3:} Example 2 illustrates the asymptotic efficiency for first-order FIR system. Actually, this conclusion can be extended to high-order FIR system by simulations. Let's come back to the third-order FIR system in Example 1.
%
We use the IMPF algorithm to estimate the unknown parameter. We repeat the simulation 500 times with the same initial estimate to establish the empirical variance of the estimation errors representing the mean square errors. The result is shown below:
\begin{figure}[htb]
	\centering \includegraphics[width=0.75\columnwidth]{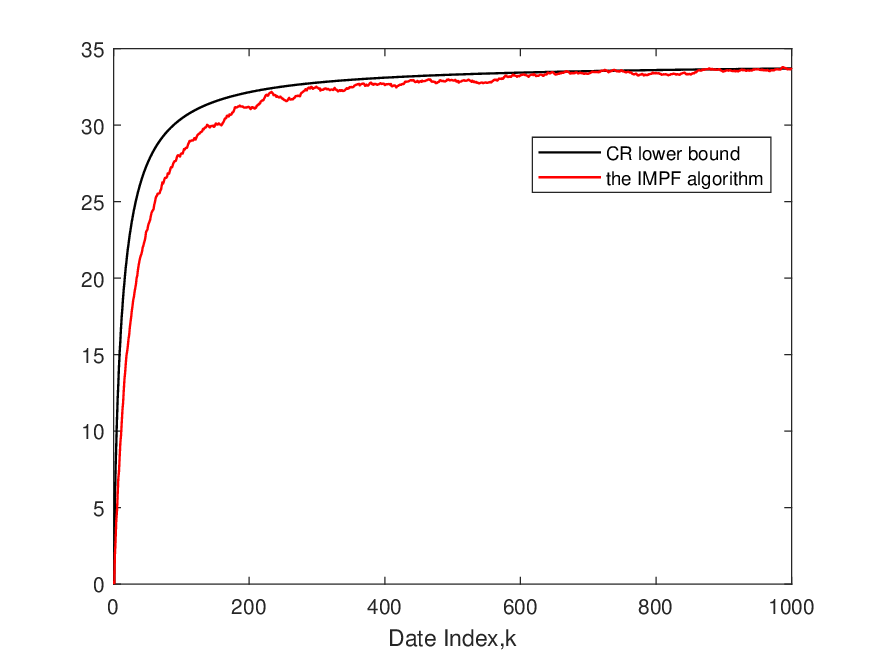}
	\caption{Comparison between the empirical variance of the IMPF algorithm and CR lower bound .}
	\label{figefficiency_high}
\end{figure}

Figure \ref{figefficiency_high} illustrates the comparison results between the empirical variance of estimation errors of the algorithm and the CR lower bound, which shows that the algorithm is asymptotically efficient. Next we will perform computational time comparison for the IMPF algorithm, the algorithm proposed in \cite{10247592}, the algorithms in \cite{wang2023asymptotically} and the algorithms in \cite{Ke2022RecursiveIO}. We evaluate the total computational time when the estimate error satisfies $E\Vert\widetilde{\theta}_k\Vert^2<10^{-4}$ for three times. 
%
\begin{table}[h]
	\centering
	\caption{\centering Computational time comparison for the IMPF algorithm \\and the algorithm in \cite{10247592}, \cite{wang2023asymptotically} and \cite{Ke2022RecursiveIO} (Time in seconds)
	}
	\begin{tabular}{l|l|l|l|l}
		\hline
		&First &Second &Third &Average\\
		\hline
		IMPF Algorithm &0.05&0.02&0.24&0.103\\  
		\hline
		Algorithm in \cite{10247592}& 65.38&45.34&93.02&67.913\\
		\hline
		Algorithm in \cite{wang2023asymptotically}& 34.85&28.39&28.73&30.657\\
		\hline
		Algorithm in \cite{Ke2022RecursiveIO}& 4.35&0.66&0.71&1.907\\
		\hline
	\end{tabular}
\end{table}

Although the algorithms in \cite{10247592} and \cite{wang2023asymptotically} are both asymptotically efficient, the time used by them are much more than the IMPF algorithm. It illustrates that the computational complexity of the IMPF algorithm has decreased a lot. Compared to the algorithm in \cite{Ke2022RecursiveIO}, the accelerated operator in IMPF algorithm can more fully utilize the prior information, so that the time used of the IMPF algorithm has decreased significantly as well.


\section{Concluding remarks}
This paper is concerned with parameter identification problem for FIR systems with binary-valued observations under low computational complexity. This paper constructed a RPFI algorithm, using a special cut-off coefficient to remove the projection operators so it reduce significantly the computational complexity. The algorithm is proved to be mean square and almost surely convergent. By designing an adaptive accelerated coefficient, the mean square convergence rate is proved to be $O\left(\frac{1}{k}\right)$, which is the same as that the system output is exactly known. Based on the RPFI algorithm and inspired by the structure of the $\mathrm{Cram\acute{e}r}$-Rao lower bound, an IMPF algorithm is constructed and it is proved to be asymptotically efficient for first-order FIR systems with proper weight coefficients.

There are many meaningful topics for future works, for example, how to give theoretical proof of the asymptotic efficiency of the proposed algorithm for the high-order FIR systems? Whether the projection-free algorithm is suitable for more complex systems with binary-valued observations, such as ARMAX systems?

\section{Appendix}
In order to analyze the convergence properties, we introduce four lemmas.
\begin{lemma} \label{lemmaconvergence}
	\cite{chen2005stochastic}Let $(v_k,\mathcal{F}_k)$ and $(w_k,\mathcal{F}_k)$ be two nonnegative adapted sequences. If $E(v_{k+1}|\mathcal{F}_k)\leq v_k+w_k$ and $E\sum_{i=1}^{\infty}w_k<\infty$, then $v_k$ converges a.s. to a finite limit.
\end{lemma}

\newtheorem{lemma22}{Lemma}[section]
\begin{lemma} \label{lemmabound}\cite{wang2023asymptotically}
	For any $a\in \mathbb{R}$ and $l\in \mathbb{Z}^+$, the following conclusions hold:
	1).The product of an infinite series satisfy the following equation:
	\begin{align*}
		\prod_{i=l+1}^{k}\left(1-\frac{a}{i}\right)=O\left(\left(\frac{l}{k}\right)^a\right),\text{ as }k\to\infty.
	\end{align*}
	2).For any $\delta>0$, we have 
	\begin{align*}
		\sum_{l=i}^{k-1}\prod_{i=l+1}^{k}\left(1-\frac{a}{i}\right)\frac{1}{l^{1+\delta}}=&
		\left\{
		\begin{aligned}
			&O\left(\frac{1}{k^\delta}\right), & \delta&<a;\\
			&O\left(\frac{\log k}{k^a}\right), & \delta&=a;\\
			&O\left(\frac{1}{k^a}\right), & \delta&>a.
		\end{aligned}
		\right.
	\end{align*}
\end{lemma}
\begin{lemma} \label{lemmaefficiency} \cite{zhang2019asymptotically}
	For the FIR system (\ref{M}) with binary-valued observations (\ref{moxingtwo}), the CR lower bound of parameter estimate is 
	\begin{align*}
		\Delta_k=\left(\sum_{i=1}^{k}\frac{f_i^2}{F_i(1-F_i)}\phi_i\phi_i^T\right)^{-1},
	\end{align*}
	where $F_i=F(C-\phi_i\theta)$ and $f_i=f(C-\phi_i\theta)$ for $i=1,2,...,k.$
\end{lemma}
\begin{lemma}\label{zenmelaide}
	Under Assumptions \ref{A1}-\ref{A3},  $\prod_{i=l}^{k}(1-P_i\beta_i\phi_i^2)= O\left(\frac{l}{k}\right)$, where $P_i$ and $\beta_i$ are defined as (\ref{Pbeta}).
	
	\begin{proof}  
		By (\ref{P_k}), we know that $P_k=O\left(\frac{1}{k}\right)$. So we have
		\begin{align*}
			\prod_{i=l}^{k}(1-P_i\beta_i\phi_i^2)=&\exp\left\{\sum_{i=l}^{k}\ln(1-P_i\beta_i\phi_i^2)\right\}\nonumber\\
			\leq&\exp\left\{\sum_{i=l}^{k}\ln\left[\left(P_i-\frac{P_i^2\beta_i\phi_i^2}{1+P_i\beta_i\phi_i^2}\right)\frac{1}{P_i}\right]\right\}\nonumber\\
			=&\exp\left\{\sum_{i=l}^{k}\ln\frac{P_{i+1}}{P_i}\right\}\nonumber\\
			=&\frac{P_{k+1}}{P_l}= O\left(\frac{l}{k}\right).
		\end{align*}
	\end{proof}
\end{lemma}


\bibliographystyle{unsrt}	
\bibliography{reference.bib}

\end{document}